%% file: main.tex
\documentclass[journal]{IEEEtran}

\PassOptionsToPackage{hyphens}{url}

\usepackage{algorithm}
\usepackage[noEnd=false,indLines=true]{algpseudocodex}
\usepackage{amsmath,amssymb,amsfonts,amsthm,mathrsfs}
\usepackage{adjustbox}
\usepackage{array}
\usepackage{booktabs}
\usepackage{cite}
\usepackage{graphicx}
\usepackage{stfloats}
\usepackage{tabularx}
\usepackage{textcomp}
\usepackage[hidelinks]{hyperref}

\hypersetup{breaklinks=true}

\input{Congmacros}

\newlength{\eqtagspace}
\newcommand{\eqfitu}[1]{%
  \adjustbox{max width=\linewidth}{\(\displaystyle #1\)}%
}

\newtheorem{proposition}{Proposition}

\newtheorem{lemma}{Lemma}

\newtheorem{remark}{Remark}

\newcommand*{\QED}{\hfill\ensuremath{\square}}

\begin{document}

\title{Cost-Informed Learning for Aggregating Building HVAC Flexibility}

\author{Jingguan Liu,~Cong Chen,~Xiaomeng Ai,~Jiakun Fang,~Jinsong Wang,~and Jinyu Wen
\thanks{This work was supported by the National Natural Science Foundation of China under Grant U25B6017.~\textit{(Corresponding authors:~Cong Chen;~Xiaomeng Ai.)}}%
\thanks{Jingguan Liu, Xiaomeng Ai, Jiakun Fang, and Jinyu Wen are with the State Key Laboratory of Advanced Electromagnetic Technology, Huazhong University of Science and Technology, Wuhan 430074, China (e-mail: \url{spencerplusmail@foxmail.com}; 
\url{xiaomengai@hust.edu.cn}; \url{jfa@hust.edu.cn}; \url{jinyu.wen@hust.edu.cn}).}%
\thanks{Cong Chen is with the Thayer School of Engineering, Dartmouth College, Hanover, NH 03755, USA (e-mail: \url{Cong.Chen@dartmouth.edu}).}%
\thanks{Jinsong Wang is with HyperStrong Technology Co., Ltd., Beijing 100094, China (e-mail: \url{wangjinsong@hyperstrong.com}).}
}

\markboth{under review}%
{Liu \MakeLowercase{\textit{et al.}}: Cost-Informed Learning for Aggregating Building HVAC Flexibility}

\maketitle

\begin{abstract}
This paper develops a cost-informed aggregation framework that learns an aggregate flexibility set of building heating, ventilation, and air-conditioning (HVAC) loads to minimize the aggregator's dispatch cost. Existing aggregation methods mainly use volume-oriented objectives and treat flexibility aggregation and downstream utilization as separate stages. Consequently, the resulting aggregate set may fail to preserve the flexibility most valuable for reducing downstream dispatch costs. To address this limitation, we represent aggregate HVAC flexibility using a parameterized storage-form surrogate and jointly learn the surrogate parameters and the inner-approximation objective from downstream dispatch-cost feedback. This cost-informed feedback allocates the surrogate's limited representation capacity to cost-relevant regions of the aggregate flexibility set. To account for electricity-price uncertainty, we formulate a distributionally robust conditional value-at-risk (DR-CVaR) downstream problem that captures both distributional ambiguity and tail risk. For efficient learning, we reformulate the DR-CVaR problem as a convex second-order cone program. We estimate cost gradients using randomized smoothing and a score-function estimator, avoiding differentiation through large-scale building-level optimization problems. Case studies using NYISO price data show that the proposed framework reduces dispatch costs relative to a volume-oriented aggregation benchmark and that suitable risk and ambiguity settings can lower the out-of-sample CVaR of dispatch costs.
\end{abstract}

\begin{IEEEkeywords}
Building HVAC systems, demand-side aggregation, cost-informed learning.
\end{IEEEkeywords}

\section{Introduction}\label{sec:intro}

\input{I_Intro}

\section{Problem Formulation}\label{sec:problem}
\input{II_Problem}

\section{Cost-Informed Learning Framework}\label{sec:framework}
\input{III_Framework}

\section{Case Studies}\label{sec:case}
\input{IV_Case}

\section{Conclusion}\label{sec:conclu}
This paper proposed a cost-informed learning framework for aggregating building HVAC flexibility. The framework learns both the storage-surrogate parameters and inner-approximation objective from downstream utilization-cost feedback, and embeds a DR-CVaR task to account for distributional ambiguity and tail risk. For HVAC aggregators, our method provides a storage-form aggregate model for downstream cost minimization, while preserving economically useful flexibility and ensuring disaggregation feasibility.

The case studies demonstrate that cost-informed learning improves downstream economic performance by preserving flexibility in cost-active directions while maintaining computational tractability. The results also show that, under unseen price-distribution shifts, the proposed risk- and ambiguity-aware formulation improves out-of-sample robustness and reduces high-cost tail outcomes for suitable risk and ambiguity settings among those tested. Overall, the proposed approach provides a principled link between physically feasible HVAC flexibility aggregation and economically effective flexibility utilization. 

Future work will focus on systematic calibration of the DR-CVaR parameters to better balance economic performance and out-of-sample robustness.

%\appendices
%\input{Appendix}

\section*{Acknowledgments}
The authors would like to thank Prof. Eilyan Bitar from Cornell University for many insightful discussions.

\bibliographystyle{IEEEtran}
\bibliography{BIB}

\clearpage
\appendices

\input{Appendix}

\end{document}

%% file: Congmacros.tex
\definecolor{revisionblue}{RGB}{0,112,192}
\newif\ifmarkrevisions
\markrevisionsfalse
\makeatletter
\def\rev@blue{0 0.43922 0.75294 rg 0 0.43922 0.75294 RG}
\def\rev@black{0 g 0 G}
\newcommand{\rev@set}[1]{%
	\let\current@color#1%
	\pdfcolorstack\@pdfcolorstack set{\current@color}%
}
\DeclareRobustCommand{\revstart}{\ifmarkrevisions\leavevmode\rev@set\rev@blue\fi}
\DeclareRobustCommand{\revend}{\ifmarkrevisions\rev@set\rev@black\fi}

\DeclareRobustCommand{\revcolor}{\ifmarkrevisions\rev@set\rev@blue\fi}
\makeatother
\DeclareRobustCommand{\rev}[1]{\ifmarkrevisions\textcolor{revisionblue}{#1}\else#1\fi}

\def\beq{\begin{equation}}
\def\eeq{\end{equation}}
\def\bea{\begin{eqnarray}}
\def\eea{\end{eqnarray}}
\def\ba{\begin{array}}
\def\ea{\end{array}}

\def\bitem{\begin{itemize}}
\def\eitem{\end{itemize}}
\def\ben{\begin{enumerate}}
\def\een{\end{enumerate}}

\def\edoc{\end{document}}

%% file: I_Intro.tex
\IEEEPARstart{C}{oordinated} operation of demand-side distributed energy resources provides an effective way to aggregate flexibility for improving power-system operational efficiency \cite{chen2023scheduling}. Building heating, ventilation, and air-conditioning (HVAC) loads are especially attractive because their fast power response and building thermal inertia allow electricity consumption to be shifted through pre-cooling or pre-heating without compromising occupant comfort \cite{Harris_2026}. Nevertheless, dispatching individual HVAC systems becomes computationally prohibitive at scale. Aggregators address this scalability issue by interfacing between system operators and HVAC fleets. They characterize the collective capability of many HVAC systems as a single aggregate flexibility set and disaggregate an approved aggregate power profile into feasible individual trajectories \cite{chen2026competitive}.

Recent market reforms, such as FERC Order No.~2222, have created pathways for HVAC aggregations to participate in power system operation and wholesale electricity markets \cite{chen2024wholesale}. In practice, an aggregator may participate through either self-dispatch or operator-dispatch modes \cite{CAISO_Sec30_2024}. Under self-dispatch, the aggregator must determine and submit a power dispatch trajectory in advance, making its market performance highly dependent on price-forecast accuracy. By contrast, under operator dispatch, the aggregator submits a compact representation of its aggregate flexibility set, such as a storage-form surrogate with time-varying power and energy limits, and the system operator determines the dispatch within the submitted set. This mode reduces the need for the aggregator to pre-commit to a fixed dispatch trajectory, but makes the quality of the submitted flexibility representation critical.

From the perspective of an HVAC aggregator, the operator-dispatch interface naturally involves two stages. In the upstream \emph{flexibility aggregation} stage, the aggregator encodes the heterogeneous thermal constraints of individual HVAC systems into a compact aggregate flexibility set for submission to the system operator \cite{liu2025leveraging}. In the downstream \emph{flexibility utilization} stage, the system operator determines an aggregate power trajectory within the submitted set under realized operating conditions and communicates the trajectory and the corresponding clearing prices to the aggregator. The aggregator then disaggregates the trajectory into feasible schedules for individual HVAC systems \cite{favaro2025decision}. Under convex market clearing with supporting prices, the dispatched trajectory also minimizes the price-taking aggregator's electricity cost over its submitted set. These two stages are tightly coupled: the aggregate set constructed in the first stage directly determines which dispatch profiles are available in the second. Therefore, an effective aggregate HVAC flexibility set must preserve the fleet's operational flexibility that is most critical for downstream utilization while remaining compatible with operator-facing dispatch interfaces. Reconciling these requirements is challenging because HVAC flexibility is governed by building thermal dynamics that differ from conventional storage dynamics \cite{wu2025energy}, making it difficult to preserve economically useful flexibility in a compact storage-form aggregate set. This motivates a cost-informed aggregation framework that explicitly links aggregate-set construction with downstream utilization.

\subsection{Related Work}
Mathematically, the exact aggregate flexibility set of an HVAC fleet is the Minkowski sum of the individual flexibility sets, for which obtaining an explicit high-dimensional multi-period representation is generally computationally intractable \cite{tiwary2008hardness}. To enable tractable computation, existing studies typically rely on either outer or inner approximations. Outer approximations, such as the direct summation of individual power envelopes in virtual-battery models \cite{song2017thermal,le2022aggregation}, yield simple aggregate representations but may overestimate the available HVAC flexibility \cite{wen2022aggregate}. Consequently, an approved aggregate dispatch may not be feasible at the individual-building level. Inner approximations are therefore widely used, as they guarantee disaggregation feasibility by construction. Representative inner-approximation templates include axis-aligned boxes \cite{chen2021leveraging}, zonotopes \cite{muller2017aggregation}, homothets \cite{hao2014aggregate,liu2024continuous}, vertex-based polytopes \cite{ozturk2024alleviating}, and affine-transformation-based polytopes \cite{al2024efficient,liu2026coupling}. Among these, affine-polytope methods, first developed in \cite{al2024efficient} and later adapted to HVAC aggregation in \cite{liu2026coupling}, have been shown to provide expressive representations of HVAC flexibility.

Despite this progress, two mismatches limit how much of the true HVAC flexibility can be converted into downstream economic value. The first is a \emph{structural mismatch} between the storage-form surrogate and the underlying HVAC flexibility. Although HVAC loads can be interpreted as thermal storage, their thermal dynamics differ from the state evolution of conventional energy storage \cite{liu2026coupling}. Mapping HVAC flexibility into a baseline-plus-storage representation can therefore discard economically useful parts of the true feasible region. The second is an \emph{objective mismatch} between the criterion used to select an inner approximation and the aggregator's actual economic objective \cite{zhu2026profit}. Once the surrogate form limits the representable flexibility, the selection objective determines which directions of the feasible region are retained. Existing inner approximation methods typically rely on geometric, volume-oriented objectives and treat aggregation and downstream utilization as separate stages. Consequently, the resulting representation may fail to preserve the flexibility most valuable for reducing dispatch costs \cite{cui2025smart}.

For a cost-minimizing aggregator, however, geometric size is not the final performance metric. Different price trajectories activate different dispatch points, so two inner approximations with similar volumes can lead to substantially different operating costs when embedded in the downstream utilization problem \cite{cui2024decision}. The economically relevant question is therefore not only how much flexibility is preserved, but also whether the preserved flexibility lies in cost-relevant regions. Historical and forecast electricity prices, which are typically available to aggregators, provide direct information about where flexibility is likely to be valuable \cite{yi2025perturbed}. To the best of our knowledge, this price information has not been systematically exploited to guide flexibility aggregation, especially when the limited representation capacity of a storage-form surrogate must be allocated across different flexibility directions. This observation motivates a cost-informed aggregation approach that uses downstream economic performance, rather than geometric size alone, to guide the design of aggregate HVAC flexibility sets.

Such a shift may offer two practical benefits. First, it avoids relying solely on high-dimensional geometric volume metrics, whose exact evaluation can be computationally challenging in high-dimensional space \cite{dong2026federated}, and instead uses dispatch cost as a direct and tractable performance signal. Second, it enables the aggregate representation to preserve flexibility in economically important regions, thereby reducing the aggregator's dispatch cost. Realizing this idea, however, raises two challenges: (i) the surrogate parameters and inner-approximation objective must be coupled with the downstream utilization so that limited representation capacity is allocated to cost-relevant flexibility regions; and (ii) because price forecasts inevitably deviate from realized prices, the learned aggregate set must generalize to unseen price scenarios and remain robust to distributional shifts and tail-risk events.

\subsection{Main Contributions}
\begin{figure}[!t]
	\centering
	\includegraphics[width=\linewidth]{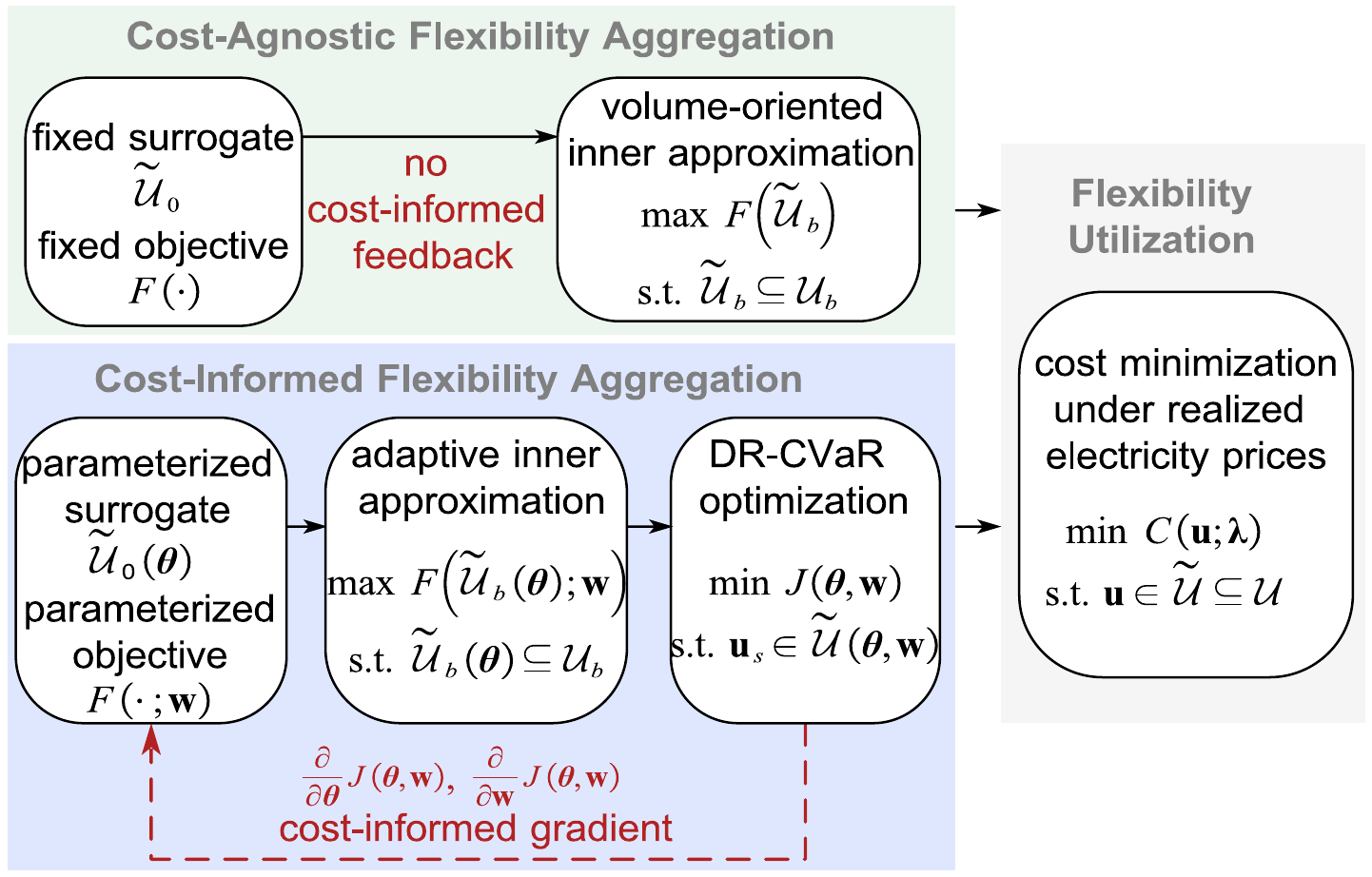}
	\caption{Comparison between cost-agnostic and cost-informed HVAC flexibility aggregation. The notations are defined in Table~\ref{tab:symbols}.  The conventional approach (top) uses a fixed storage-form surrogate and a fixed volume-oriented inner-approximation objective, without any feedback from downstream utilization costs. The proposed approach (bottom) parameterizes both the storage-form surrogate in Eq. \eqref{eq:surrogate_storage_model_compact} and the inner-approximation objective $F(\cdot;\mathbf{w})$ in Eq. \eqref{eq:innerapp_problem_original}, and updates them using cost-informed feedback from the downstream DR-CVaR utilization layer in Eq. \eqref{eq:DRCVAR_problem_SOCP}. The learned aggregate set is then used for realized-price cost minimization in Eq. \eqref{eq:utilize_realizedprice}.}
	\label{fig:Overview}
\end{figure}

\begin{table}[!t]
	\footnotesize
	\caption{Main symbols used in Fig.~\ref{fig:Overview}}
	\begin{center}
		\begin{tabularx}{\linewidth}{lX}
			\hline
			Symbol & Description \\
			\hline
			$\mathcal{U}_b$ & Individual HVAC flexibility set of building $b$. \\
			$\mathcal{U}$ & Exact aggregate HVAC flexibility set. \\
			$\mathcal{U}_0(\boldsymbol{\theta})$ & Parameterized storage-form surrogate \\
			$\boldsymbol{\theta}$ & Learnable storage-surrogate parameter vector. \\
			$\widetilde{\mathcal{U}}_b(\boldsymbol{\theta})$ & Building-level inner-approximation set. \\
			$\widetilde{\mathcal{U}}(\boldsymbol{\theta},\mathbf{w})$ & Aggregate inner-approximation set. \\
			$\mathbf{w}$ & Learnable weight vector for inner approximation. \\
			$F(\cdot;\mathbf{w})$ & Shorthand for the weighted inner-approximation objective in \eqref{eq:innerapp_problem_original}. \\
			$J(\boldsymbol{\theta},\mathbf{w})$ & DR-CVaR utilization cost. \\
			$C(\mathbf{u};\boldsymbol{\lambda})$ & Realized electricity-cost function for aggregate dispatch $\mathbf{u}$ under price $\boldsymbol{\lambda}$. \\
			$\mathbf{u}_s$ & Aggregate dispatch associated with price scenario $s$. \\
			$\partial J/\partial\boldsymbol{\theta},\,\partial J/\partial\mathbf{w}$ & Cost-informed gradients used to update $\boldsymbol{\theta}$ and $\mathbf{w}$. \\
			\hline
		\end{tabularx}
		\label{tab:symbols}
	\end{center}
\end{table}

This paper proposes a \emph{cost-informed learning framework} for aggregating HVAC flexibility. Inspired by recent decision-focused learning approaches that often train predictive components using downstream optimization performance \cite{li2016toward,han2021task,cui2024decision}, our framework does not learn a prediction model, but instead learns a set-valued aggregate flexibility representation. Specifically, we extend the existing cost-agnostic aggregation method in \cite{al2024efficient} by using downstream utilization-cost feedback to jointly learn the surrogate parameters and inner-approximation objective. This closes the loop between flexibility aggregation and downstream utilization, as illustrated in Fig.~\ref{fig:Overview} (the main notation is summarized in Table~\ref{tab:symbols}). 

Our main contributions are threefold:

\subsubsection{Framework}
We develop a novel cost-informed learning framework that couples aggregate HVAC flexibility identification with downstream flexibility utilization. Instead of constructing a volume-oriented aggregate set, our closed-loop formulation adaptively learns both the surrogate parameters and inner-approximation objective from utilization-side cost feedback. This directs the approximate surrogate's representation capacity toward HVAC flexibility regions identified as cost-relevant by the downstream objective.
	
\subsubsection{Modeling}
We formulate a risk- and ambiguity-aware flexibility utilization task using historical or forecast price trajectories available to HVAC aggregators. This task adopts a distributionally robust conditional value-at-risk (DR-CVaR) objective to capture price-distribution ambiguity and tail-risk preferences. Embedding this task into the proposed learning pipeline enables the learned aggregate set to reduce out-of-sample tail costs under unseen price scenarios.
	
\subsubsection{Solution}
We design a tailored solution procedure for tractable calculation. The DR-CVaR utilization layer is reformulated via duality into a convex second-order cone program (SOCP), while a score-function gradient estimator with randomized perturbations is used to estimate cost-informed gradients efficiently. The separability of individual flexibility sets across buildings further enables parallel inner approximation and gradient evaluation for efficient learning.

\subsection{Paper Organization and Notations}
The remainder of this paper is organized as follows. Section \ref{sec:problem} formulates the individual and aggregate HVAC flexibility models and motivates cost-informed aggregation. Section \ref{sec:framework} presents the proposed cost-informed learning framework, including the parameterized
surrogate, inner approximation, utilization task, and solution workflow. Section \ref{sec:case} reports case studies, and Section \ref{sec:conclu} concludes the paper.

We use the following notation. The set of real numbers is denoted by $\mathbb{R}$. For any positive integer $n$, $\mathbf{I}_n$, $\mathbf{1}_n$, and $\mathbf{0}_n$ denote the $n \times n$ identity matrix, the $n \times 1$ all-ones vector, and the $n \times 1$ all-zeros vector, respectively. For a vector $\boldsymbol{\gamma}$ and a matrix $\boldsymbol{\Gamma}$ with compatible dimensions, the affine image of a set $\mathcal{X}$ is defined as $\boldsymbol{\gamma }+\boldsymbol{\Gamma }\mathcal{X} :=\left\{ \boldsymbol{\gamma }+\boldsymbol{\Gamma x}\mid \mathbf{x}\in \mathcal{X} \right\} $. For a scalar $x \in \mathbb{R}$, $(x)_+:=\max \left\{ 0,x \right\} $ denotes the positive part of $x$. The notation $\left\| \cdot \right\| $ denotes the Euclidean norm.

%% file: II_Problem.tex
In this section, we first model individual and aggregate HVAC flexibility sets and then motivate cost-informed aggregate flexibility identification.

\subsection{Individual Building HVAC Flexibility Modeling}\label{subsec:problem_individual}
We consider a fleet of $B$ buildings indexed by $b\in \mathcal{B}:=\left\{1,\ldots,B\right\}$. Time is discretized into $T$ equal-length periods of duration $\delta$, indexed by $t\in \mathcal{T}:=\left\{1,\ldots,T\right\}$. The thermal dynamics of each building are modeled by a second-order finite-difference formulation \cite{zheng2022consensus,liu2025leveraging}:
\beq \label{eq:individual_dynamics_I}
\eqfitu{
c_{b}^{\mathrm{I}}\frac{x_{b,t}^{\mathrm{I}}-x_{b,t-1}^{\mathrm{I}}}{\delta}=\frac{x_{t}^{\mathrm{O}}-x_{b,t}^{\mathrm{I}}}{r_{b}^{\mathrm{I}}}+\frac{x_{b,t}^{\mathrm{M}}-x_{b,t}^{\mathrm{I}}}{r_{b}^{\mathrm{M}}}-\left(1-\alpha_b\right)\eta_bu_{b,t},
}
\eeq
\beq \label{eq:individual_dynamics_M}
\eqfitu{
c_{b}^{\mathrm{M}}\frac{x_{b,t}^{\mathrm{M}}-x_{b,t-1}^{\mathrm{M}}}{\delta}=\frac{x_{b,t}^{\mathrm{I}}-x_{b,t}^{\mathrm{M}}}{r_{b}^{\mathrm{M}}}-\alpha_b\eta_b u_{b,t},
}
\eeq
where the decision variable $u_{b,t}$ denotes HVAC power, $x_{b,t}^{\mathrm{M}}$ is the building mass temperature, and $x_{b,t}^{\mathrm{I}}$ is the indoor air temperature. The parameters $r_{b}^{\mathrm{I}}$ and $r_{b}^{\mathrm{M}}$ are thermal resistances, $c_{b}^{\mathrm{I}}$ and $c_{b}^{\mathrm{M}}$ are thermal capacitances, $\eta_b$ is the coefficient of performance, $x_t^{\mathrm{O}}$ is the outdoor temperature, and $\alpha_b$ is the fractional coefficient of HVAC power.

To ensure users' thermal comfort and HVAC operating limits, we impose the indoor temperature and power constraints
\beq \label{eq:individual_limit_x}
\eqfitu{
x_{b}^{\mathrm{set}}-x_{b}^{\mathrm{tol}}\le x_{b,t}^{\mathrm{I}}\le x_{b}^{\mathrm{set}}+x_{b}^{\mathrm{tol}},
}
\eeq
\beq \label{eq:individual_limit_u}
\eqfitu{
\underline{u}_b\le u_{b,t}\le \overline{u}_b,
}
\eeq
where $x_{b}^{\mathrm{set}}$ is the temperature setpoint, $x_{b}^{\mathrm{tol}}$ is the comfort tolerance band, and $\underline{u}_{b}$ and $\overline{u}_{b}$ are the minimum and maximum HVAC operating powers.

By combining \eqref{eq:individual_dynamics_I}-\eqref{eq:individual_limit_u} and eliminating the temperature states $x_{b,t}^{\mathrm{M}}$ and $x_{b,t}^{\mathrm{I}}$, the feasible region of each HVAC system can be encoded as an individual flexibility set in the power subspace:
\beq \label{eq:individual_set}
\eqfitu{
\mathcal{U}_b:=\left\{ \mathbf{u}_b\, \middle| \,\mathbf{H}_b\mathbf{u}_b\le \mathbf{h}_b \right\} \subset \mathbb{R}^T,
}
\eeq
where $\mathbf{u}_b:=\left[u_{b,t}\right]_{t\in\mathcal{T}}$, $\mathbf{H}_b$ is the coefficient matrix, and $\mathbf{h}_b$ is the coefficient vector. The detailed elimination process can be found in \cite[Appendix B]{liu2025leveraging} and is omitted here for brevity.

\subsection{Flexibility Aggregation}\label{subsec:problem_aggregation}
For an aggregator controlling $B$ buildings, the aggregate flexibility set describes all feasible aggregate power profiles and is represented by the Minkowski sum of the individual HVAC flexibility sets:
\beq \label{eq:aggregate_exact_set}
\eqfitu{
\mathcal{U}:=\bigoplus_{b\in\mathcal{B}}\mathcal{U}_b=\left\{ \mathbf{u}\, \middle| \,\mathbf{u}=\sum_{b\in\mathcal{B}}\mathbf{u}_b,\mathbf{u}_b\in\mathcal{U}_b \right\} \subset \mathbb{R}^T,
}
\eeq
where $\bigoplus$ denotes the Minkowski sum.

In general, obtaining an explicit representation of the Minkowski sum is computationally intractable for high-dimensional multi-period sets \cite{tiwary2008hardness}. Therefore, we seek an inner-approximated aggregate set $\widetilde{\mathcal{U}}\subseteq\mathcal{U}$ that ensures disaggregation feasibility and computational efficiency.

\subsection{Flexibility Utilization}\label{subsec:problem_utilization}
Given an identified inner approximation $\widetilde{\mathcal{U}}$, the submitted set is dispatched according to the following price-taking proxy, which evaluates the economic value of operator dispatch under an exogenous realized electricity price vector $\boldsymbol{\lambda}\in\mathbb{R}^T$ \cite{al2024efficient}. The realized utilization value is

\beq \label{eq:utilize_realizedprice}
\eqfitu{
V\left( \boldsymbol{\lambda};\widetilde{\mathcal{U}} \right) :=\min_{\mathbf{u}\in \widetilde{\mathcal{U} }} \,\,C\left( \mathbf{u};\boldsymbol{\lambda } \right) ,\quad C\left( \mathbf{u};\boldsymbol{\lambda } \right) :=\delta \boldsymbol{\lambda }^{\top}\mathbf{u}.
}
\eeq

From an economic perspective, flexibility aggregation for HVAC aggregators should preserve the flexibility that is valuable for the downstream utilization problem \eqref{eq:utilize_realizedprice}, rather than only construct a volume-oriented surrogate set. This distinction is important for inner approximation because some flexibility loss is inevitable. Moreover, not all regions of the aggregate flexibility set are equally cost-relevant: different price trajectories activate different dispatch points, and surrogate sets with the same approximation volume may lead to substantially different operating costs. This motivates us to learn $\widetilde{\mathcal{U}}$ using utilization-side cost feedback, so that its representation capacity is concentrated on cost-relevant HVAC flexibility regions.

Realizing this goal raises three issues: how to incorporate downstream cost feedback into inner approximation, how to formulate the utilization task under price uncertainty so that $\widetilde{\mathcal{U}}$ generalizes to unseen price scenarios, and how to solve the learning pipeline efficiently. To address these issues, we next present a cost-informed learning framework.

%% file: III_Framework.tex
In this section, we present a cost-informed learning framework for characterizing aggregate HVAC flexibility. By coupling aggregate flexibility characterization with downstream utilization costs, the framework jointly learns surrogate parameters and inner-approximation objective. We first introduce the parameterized surrogate and its inner-approximation model, then formulate the risk- and ambiguity-aware utilization problem, and finally present the closed-loop learning formulation, solution procedure, and implementation workflow.

\subsection{Parameterized Storage Surrogate}\label{subsec:framework_surrogate}
To structure the inner-approximation set $\widetilde{\mathcal{U}}\subseteq\mathcal{U}$, we introduce a parameterized storage surrogate $\mathcal{U}_0\subset \mathbb{R}^{2T}$:
\beq \label{eq:surrogate_storage_model}
\eqfitu{
\mathcal{U} _0:=\left\{ \rev{\left[\left(\mathbf{u}^{\mathrm{C}}\right)^\top,\left(\mathbf{u}^{\mathrm{D}}\right)^\top\right]^\top} \, \middle| \,\begin{array}{c}
	\underline{x}\le \kappa ^tx_0+\sum_{\iota =1}^t{\kappa ^{t-\iota}\left[ \delta \left( u_{\iota}^{\mathrm{C}}\eta ^{\mathrm{C}}-u_{\iota}^{\mathrm{D}}/\eta ^{\mathrm{D}} \right) \right]}\le \overline{x},\\
	0\le u_{t}^{\mathrm{C}}\le \overline{u}^{\mathrm{C}},\;0\le u_{t}^{\mathrm{D}}\le \overline{u}^{\mathrm{D}},\forall t\in \mathcal{T}\\
\end{array} \right\} ,
}
\eeq
The first constraint in \eqref{eq:surrogate_storage_model} enforces the storage-state bounds induced by the cumulative charging/discharging trajectory, while the second and third constraints impose the charging and discharging power limits, respectively.
Here, decision variables $\mathbf{u}^{\mathrm{C}}:=\left[ u_{t}^{\mathrm{C}} \right] _{t\in \mathcal{T}}$ and $\mathbf{u}^{\mathrm{D}}:=\left[ u_{t}^{\mathrm{D}} \right] _{t\in \mathcal{T}}$ collect the charging and discharging power across all periods. The model parameters $\underline{x}$ and $\overline{x}$ are the energy lower and upper bounds, $\kappa$ is the state-retention coefficient, $x_0$ is the initial state before the first dispatch period, $\eta^{\mathrm{C}}$ and $\eta^{\mathrm{D}}$ are the charging and discharging efficiencies, and $\overline{u}^{\mathrm{C}}$ and $\overline{u}^{\mathrm{D}}$ are the charging and discharging power capacities. The storage surrogate is parameterized by
\beq \label{eq:surrogate_storage_parameter}
\eqfitu{
\boldsymbol{\theta } :=\left[ x_0,\underline{x},\overline{x},\kappa ,\eta ^{\mathrm{C}},\eta ^{\mathrm{D}},\overline{u}^{\mathrm{C}},\overline{u}^{\mathrm{D}} \right] ^{\top}\in \mathbb{R} ^8 .
}
\eeq
Then, we can rewrite the storage surrogate into a compact parameterized form
\beq \label{eq:surrogate_storage_model_compact}
\eqfitu{
\mathcal{U} _0\left( \boldsymbol{\theta } \right) =\left\{ \rev{\left[\left(\mathbf{u}^{\mathrm{C}}\right)^\top,\left(\mathbf{u}^{\mathrm{D}}\right)^\top\right]^\top} \, \middle| \,\mathbf{H}_0\left( \boldsymbol{\theta } \right) \rev{\left[\left(\mathbf{u}^{\mathrm{C}}\right)^\top,\left(\mathbf{u}^{\mathrm{D}}\right)^\top\right]^\top} \,\le \mathbf{h}_0\left( \boldsymbol{\theta } \right) \right\} .
}
\eeq
where $\mathbf{H}_0\left( \boldsymbol{\theta } \right) $ and $\mathbf{h}_0\left( \boldsymbol{\theta } \right) $ are the coefficient matrix and vector derived from \eqref{eq:surrogate_storage_model}, respectively.

The storage-form surrogate is adopted for two reasons. First, storage-form flexibility is physically interpretable and compatible with many market and system operation requirements \cite{CAISO_Sec30_2024}, enabling the aggregate set to be readily usable. Second, HVAC flexibility can be interpreted as a form of thermal energy storage \cite{song2017thermal}: increasing cooling power charges thermal flexibility, whereas decreasing cooling power discharges it. Therefore, a storage-form surrogate with properly selected parameters can provide an expressive approximation of aggregate HVAC flexibility. However, due to the structural mismatch between HVAC thermal dynamics and storage parameters, selecting these parameters heuristically is challenging. We therefore treat $\boldsymbol{\theta}$ as learnable rather than fixed, and update it using downstream cost feedback.

\subsection{Parameterized Inner Approximation}\label{subsec:framework_innerapp}
We next construct  $\widetilde{\mathcal{U}}\left( \boldsymbol{\theta } \right)$ as an affine image of the parameterized storage surrogate. This affine form follows existing inner-approximation methods \cite{al2024efficient,liu2026coupling} and provides a tractable representation for aggregating many HVAC flexibility sets. Specifically, for individual flexibility set $\mathcal{U}_b$, we have
\beq \label{eq:innerapp_individual_affine}
\eqfitu{
\widetilde{\mathcal{U} }_b\left( \boldsymbol{\theta } \right) :=\boldsymbol{\gamma }_b+\boldsymbol{\Phi}\boldsymbol{\Gamma }_b\mathcal{U} _0\left( \boldsymbol{\theta } \right) \subseteq \mathcal{U} _b,
}
\eeq
where $\boldsymbol{\gamma }_b\in \mathbb{R} ^T$ denotes the translation vector, $\boldsymbol{\Gamma }_b\in \mathbb{R} ^{2T\times 2T}$ denotes a diagonal linear transformation matrix, and $\boldsymbol{\Phi }:=\left[ \mathbf{I}_T,-\mathbf{I}_T \right] \in \mathbb{R} ^{T\times 2T}$ denotes the projection matrix that maps bidirectional charging and discharging power to net power.

Then, the inner-approximated aggregate set can be efficiently calculated by \cite{al2024efficient}:
\beq \label{eq:innerapp_aggregate_affine}
\eqfitu{
\widetilde{\mathcal{U} }:=\left( \sum_{b\in \mathcal{B}}{\boldsymbol{\gamma }_b} \right) +\boldsymbol{\Phi }\left( \sum_{b\in \mathcal{B}}{\boldsymbol{\Gamma }_b} \right) \mathcal{U} _0\left( \boldsymbol{\theta } \right)
\subseteq\mathcal{U}.
}
\eeq
\begin{remark}[Diagonal linear transformation]\label{remark:diag_trans}
Restricting $\boldsymbol{\Gamma}_b$ to be diagonal preserves the storage-form interpretability of the resulting inner approximation while improving computational and learning efficiency. Structurally, a diagonal $\boldsymbol{\Gamma}_b$ independently rescales the charging and discharging coordinates of $\mathcal{U}_0(\boldsymbol{\theta})$ at each time period. This time-varying scaling changes how the storage surrogate is embedded into each building's flexibility set, but it does not mix time periods or charging/discharging directions. Therefore, the transformed set remains interpretable as a baseline profile plus storage-form flexibility with time-varying charging/discharging scaling. Computationally, the diagonal restriction uses only $2T$ linear-map variables instead of $4T^2$ entries of a full $2T\times 2T$ transformation, reducing the size of the building-level inner-approximation problem \eqref{eq:innerapp_problem_linear} and keeping the cost-informed weight vector $\mathbf{w}$ low-dimensional for efficient learning.
\end{remark}

To determine the suitable value of affine parameters $\boldsymbol{\Gamma }_b$ and $\boldsymbol{\gamma }_b$ in \eqref{eq:innerapp_individual_affine}, we formulate a parameterized inner approximation optimization problem for each building:
\beq \label{eq:innerapp_problem_original}
\eqfitu{
\begin{aligned}
	\max_{\boldsymbol{\gamma }_b,\boldsymbol{\Gamma }_b} \quad &\mathbf{w}^{\top}\rev{\left[\left(\mathrm{diag}\left( \boldsymbol{\Gamma }_b \right)\right)^\top,\left(\boldsymbol{\gamma }_b\right)^\top\right]^\top}\\
	\mathrm{s}.\mathrm{t}.\quad &\widetilde{\mathcal{U} }_b\left( \boldsymbol{\theta } \right) =\boldsymbol{\gamma }_b+\boldsymbol{\Phi}\boldsymbol{\Gamma }_b\mathcal{U} _0\left( \boldsymbol{\theta } \right) \subseteq \mathcal{U} _b,\\
\end{aligned}
}
\eeq
where the objective is a weighted surrogate-size objective, $\mathbf{w}\in \mathbb{R} ^{3T}$ is the weight vector, and $\mathrm{diag}\left( \boldsymbol{\Gamma }_b \right) $ collects the diagonal entries of $\boldsymbol{\Gamma }_b$. For compact notation, the first-line objective in \eqref{eq:innerapp_problem_original} is denoted by $F(\widetilde{\mathcal{U}}_b(\boldsymbol{\theta});\mathbf{w})$ in Fig.~\ref{fig:Overview}. Further using \cite[Lemma 1]{al2024efficient}, \eqref{eq:innerapp_problem_original} can be equivalently expressed as a parameterized linear program (LP):
\beq \label{eq:innerapp_problem_linear}
\eqfitu{
\begin{aligned}
	\max_{\boldsymbol{\gamma }_b,\boldsymbol{\Gamma }_b,\boldsymbol{\Lambda }_b} \quad &\mathbf{w}^{\top}\rev{\left[\left(\mathrm{diag}\left( \boldsymbol{\Gamma }_b \right)\right)^\top,\left(\boldsymbol{\gamma }_b\right)^\top\right]^\top}\\
	\mathrm{s}.\mathrm{t}.\quad &\boldsymbol{\Lambda }_b\ge 0,\boldsymbol{\Lambda }_b\mathbf{H}_0\left( \boldsymbol{\theta } \right) =\mathbf{H}_b\boldsymbol{\Phi}\boldsymbol{\Gamma }_b,\\
	&\boldsymbol{\Lambda }_b\mathbf{h}_0\left( \boldsymbol{\theta } \right) \le \mathbf{h}_b-\mathbf{H}_b\boldsymbol{\gamma }_b.\\
\end{aligned}
}
\eeq

\begin{remark}[Comparison to existing volume-oriented methods]\label{remark:compare_bitar}
If $\boldsymbol{\theta}$ is fixed and $\mathbf{w}=\rev{\left[\left(\mathbf{1}_{2T}\right)^\top,\left(\mathbf{0}_T\right)^\top\right]^\top}$, the proposed parameterized formulation reduces to the existing volume-oriented inner-approximation method in \cite{al2024efficient,liu2026coupling}. 
While the recent method in \cite{dong2026federated} optimizes the surrogate parameters to improve the inner approximation, it is still driven by a volume-oriented criterion.
Direct volume evaluation can be computationally demanding
for high-dimensional projected sets, such as the $T=24$
aggregate set considered in Section~\ref{sec:case}.
Although \cite{dong2026federated} develops a tailored
volume-evaluation procedure that exploits its specific base-set
structure, this procedure does not directly apply to the
projected aggregate set in
\eqref{eq:innerapp_aggregate_affine}.
More importantly, these geometric objectives do not explicitly
account for downstream utilization costs, and a larger-volume
approximation does not necessarily yield a lower dispatch cost.
In contrast, the proposed cost-informed framework avoids
explicit volume evaluation and jointly learns
$\boldsymbol{\theta}$ and $\mathbf{w}$ from downstream cost
feedback, enabling the aggregate surrogate to allocate its
limited representation capacity to cost-relevant
flexibility regions.
\end{remark}

\subsection{Risk- and Ambiguity-Aware Flexibility-Utilization Task}\label{subsec:framework_task}
Given surrogate parameters $\boldsymbol{\theta}$ and objective weight $\mathbf{w}$, solving \eqref{eq:innerapp_problem_linear} for all buildings and aggregating the resulting inner approximations through \eqref{eq:innerapp_aggregate_affine} gives an approximate  aggregate set $\widetilde{\mathcal{U}}$. In the sequel, we write this set as $\widetilde{\mathcal{U}}\left( \boldsymbol{\theta},\mathbf{w} \right)$ when emphasizing its dependence on the physical storage parameters and objective weights. If the realized price vector $\boldsymbol{\lambda}$ were known before aggregation, the downstream utilization value could be evaluated directly by \eqref{eq:utilize_realizedprice}. In practice, however, flexibility aggregation is made before prices are realized. Therefore, the aggregator must rely on historical or forecast price trajectories to evaluate the distribution of realized utilization values under price uncertainty.

We consider a set of $S$ known price trajectories
$\{\widetilde{\boldsymbol{\lambda}}_s\}_{s\in\mathcal{S}}$,
where $\mathcal{S}:=\{1,\ldots,S\}$, obtained from historical data or forecasts available to the aggregator. These trajectories induce the finitely supported nominal distribution $P_0$, and we write $\widetilde{\boldsymbol{\lambda}}\sim P_0$ for the nominal empirical price trajectory. Since the true distribution $P$ of the realized price vector $\boldsymbol{\lambda}$ is unknown, we assume that $P$ lies in an ambiguity set $\mathcal{P}$ centered at $P_0$, defined via the 2-Wasserstein distance:
\beq \label{eq:DR_ambiguityset}
\eqfitu{
\mathcal{P} :=\left\{ P\mid W_2(P,P_0)\le r \right\} ,
}
\eeq
\beq \label{eq:DR_Wdistance}
\eqfitu{
W_2(P,P_0)^2:=\min_{\pi \in \Pi (P,P_0)}
\int_{\mathbb{R} ^T\times \mathbb{R} ^T}
{\left\| \boldsymbol{\lambda }-\widetilde{\boldsymbol{\lambda}} \right\|_2^2\,
\pi \left(\mathrm{d}\boldsymbol{\lambda},\mathrm{d}\widetilde{\boldsymbol{\lambda}}\right)},
}
\eeq
where $\Pi (P,P_0)$ denotes the set of all joint distributions with marginal distributions $P$ and $P_0$. The radius $r$ reflects the aggregator's confidence in the nominal distribution: a larger $r$ allows more distributional deviation, while $r=0$ recovers the nominal distribution $P_0$.

Having specified the distributional ambiguity model, we adopt conditional value-at-risk (CVaR) to capture tail-risk preference. At tail probability level $\alpha \in \left( 0,1 \right]$, for any fixed aggregate set $\widetilde{\mathcal{U}}$, the CVaR of the random utilization value $V\left( \boldsymbol{\lambda};\widetilde{\mathcal{U}} \right)$ under distribution $P$ is defined as
\beq \label{eq:CVAR_def}
\eqfitu{
\mathrm{CVaR}_{P}^{\alpha}\left( V\left( \boldsymbol{\lambda};\widetilde{\mathcal{U}} \right) \right)
:=\min_{\tau \in \mathbb{R}} \left\{ \tau +\frac{1}{\alpha}\mathbb{E} _P\left[ \left( V\left( \boldsymbol{\lambda};\widetilde{\mathcal{U}} \right)-\tau \right) _+ \right] \right\} .
}
\eeq

Accordingly, for the parameterized aggregate set $\widetilde{\mathcal{U}}\left( \boldsymbol{\theta},\mathbf{w} \right)$, the risk- and ambiguity-aware utilization objective is the worst-case CVaR over all distributions in $\mathcal{P}$:
\beq \label{eq:DRCVAR_problem_original}
\eqfitu{
J\left( \boldsymbol{\theta},\mathbf{w} \right) :=\max_{P\in \mathcal{P}} \,\,\mathrm{CVaR}_{P}^{\alpha}\left( V\left( \boldsymbol{\lambda};\widetilde{\mathcal{U}}\left( \boldsymbol{\theta},\mathbf{w} \right) \right) \right) .
}
\eeq

This DR-CVaR formulation includes several useful special cases. When $\alpha =1$, CVaR coincides with expected cost, and \eqref{eq:DRCVAR_problem_original} reduces to distributionally robust expected-value evaluation of price-adaptive dispatch. As $\alpha \downarrow 0$, CVaR approaches the worst-case realized cost, so the problem approaches robust cost evaluation. When $r=0$, the ambiguity set collapses to $P_0$, and \eqref{eq:DRCVAR_problem_original} becomes a nominal CVaR evaluation problem. Therefore, the proposed task allows the aggregator to tune both distributional ambiguity through $r$ and tail-risk preference through $\alpha$. Embedding this task into the learning pipeline exposes aggregate-set learning to distributional deviations and high-cost tail events represented by the DR-CVaR objective.

\subsection{Closed-Loop Formulation with Cost-Informed Learning}\label{subsec:framework_closedloop}
With the preceding components in place, we now formulate the closed-loop learning problem as
\beq \label{eq:closedloop_original}
\eqfitu{
\begin{aligned}
	\min_{\boldsymbol{\theta},\mathbf{w}} \quad &J\left( \boldsymbol{\theta},\mathbf{w} \right)\\
	\mathrm{s}.\mathrm{t}.\quad &\widetilde{\mathcal{U}}\left( \boldsymbol{\theta},\mathbf{w} \right)=\left( \sum_{b\in \mathcal{B}}{\boldsymbol{\gamma }_b^{\star}} \right) +\boldsymbol{\Phi }\left( \sum_{b\in \mathcal{B}}{\boldsymbol{\Gamma }_b^{\star}} \right) \mathcal{U}_0\left( \boldsymbol{\theta} \right),\\
	&\left( \boldsymbol{\gamma }_b^{\star},\boldsymbol{\Gamma }_b^{\star},\boldsymbol{\Lambda }_b^{\star} \right) \in \Omega_b\left( \boldsymbol{\theta},\mathbf{w} \right),\quad \forall b\in \mathcal{B}.
\end{aligned}
}
\eeq
where $\Omega_b\left( \boldsymbol{\theta},\mathbf{w} \right)$ denotes the optimal-solution set of the building-level inner-approximation problem \eqref{eq:innerapp_problem_linear}, and the starred variables denote selected lower-level optimal solutions. Problem \eqref{eq:closedloop_original} makes the cost-informed coupling explicit: $\boldsymbol{\theta}$ defines the storage surrogate, $\mathbf{w}$ selects the inner-approximation directions through the weighted objective, and the resulting affine mappings determine the aggregate set $\widetilde{\mathcal{U} }\left( \boldsymbol{\theta },\mathbf{w} \right) $ used by the upper DR-CVaR utilization layer. Therefore, the utilization cost can provide cost-informed feedback for learning both the surrogate parameters $\boldsymbol{\theta}$ and the objective weight $\mathbf{w}$.

We then solve \eqref{eq:closedloop_original} by gradient descent, as illustrated in Fig. \ref{fig:Overview} (bottom). 
In the forward pass at each iteration, the current $\boldsymbol{\theta}$ and $\mathbf{w}$ are used to solve the building-level LPs \eqref{eq:innerapp_problem_linear}, construct $\widetilde{\mathcal{U}}(\boldsymbol{\theta},\mathbf{w})$ via \eqref{eq:innerapp_aggregate_affine}, and evaluate the DR-CVaR utilization cost. 
In the backward pass, the gradients of this cost with respect to $\boldsymbol{\theta}$ and $\mathbf{w}$ are then estimated and used to update the parameters toward lower downstream cost. 
This process is repeated until convergence.

However, two computational challenges remain. First, the DR-CVaR formulation in \eqref{eq:DRCVAR_problem_original} involves an infinite-dimensional worst-case distributional evaluation. Second, differentiating through the large-scale parameterized LPs in \eqref{eq:innerapp_problem_linear} can be computationally expensive and numerically sensitive to degeneracy. These challenges motivate the tailored solution procedure developed next.
\subsection{Tailored and Computationally Efficient Solution Procedure}\label{subsec:framework_solution}
\subsubsection{Convex Reformulation in Forward Pass}
We first reformulate the DR-CVaR evaluation in \eqref{eq:DRCVAR_problem_original} as a tractable second-order cone program (SOCP). For fixed $\boldsymbol{\theta}$ and $\mathbf{w}$, the building-level LPs in \eqref{eq:innerapp_problem_linear} define the aggregate set $\widetilde{\mathcal{U}}(\boldsymbol{\theta},\mathbf{w})$. The remaining task is to evaluate the worst-case CVaR over the Wasserstein ambiguity set. By applying Wasserstein-CVaR duality, we first obtain the following dual reformulation.

\begin{lemma}[Dual Reformulation]\label{lemma:DRCVAR_dual}
	For any given $\boldsymbol{\theta}$, $\mathbf{w}$, and $r>0$, the worst-case CVaR value in \eqref{eq:DRCVAR_problem_original} can be written as
	\beq \label{eq:DRCVaR_dual_formulation} 
	\eqfitu{
		J\left( \boldsymbol{\theta},\mathbf{w} \right)=\min_{\tau \in \mathbb{R} ,\,\mu \ge 0} \left\{ \tau +\frac{1}{\alpha}\left( \mu r^2+\mathbb{E} _{P_0}\left[ \phi \left(\tau ,\mu ,\widetilde{\boldsymbol{\lambda}}\right) \right] \right) \right\} ,
	}
	\eeq	
	where $\tau$ is the CVaR threshold variable, $\mu$ is the nonnegative dual multiplier associated with the Wasserstein ambiguity-radius constraint, and
	\beq \label{eq:DRCVaR_dual_definition} 
	\eqfitu{
		\phi \left(\tau ,\mu ,\widetilde{\boldsymbol{\lambda}}\right):=\max_{\boldsymbol{\lambda }\in \mathbb{R} ^T} \left\{ \left( V\left( \boldsymbol{\lambda};\widetilde{\mathcal{U}}\left( \boldsymbol{\theta},\mathbf{w} \right) \right) -\tau \right) _+-\mu \left\| \boldsymbol{\lambda }-\widetilde{\boldsymbol{\lambda}} \right\| _{2}^{2} \right\} .
	}
	\eeq	
\end{lemma}

Lemma~\ref{lemma:DRCVAR_dual} directly specializes the generic Wasserstein strong-duality result in \cite[Theorem 5.22]{kuhn2025distributionally}. We then combine this lemma with the finite-support reformulation developed in \cite[Theorem 3]{taha2026wasserstein}, and exploit the price-adaptive structure of $V$ to derive a finite-dimensional conic reformulation.

\begin{proposition}[SOCP Reformulation]\label{prop:DRCVAR_problem_SOCP}
Given $\boldsymbol{\theta}$ and $\mathbf{w}$ such that $\widetilde{\mathcal{U}}(\boldsymbol{\theta},\mathbf{w})$ is nonempty, compact, and convex, suppose that $r>0$ and that the nominal distribution $P_0$ is uniform over $\{\widetilde{\boldsymbol{\lambda}}_s\}_{s\in\mathcal{S}}$. Then, the forward DR-CVaR utilization layer can be equivalently reformulated as the following SOCP:
\beq \label{eq:DRCVAR_problem_SOCP} 
\eqfitu{
\begin{aligned}
	\min_{\tau ,\,\mu ,\,\{\rho _s,\zeta_s,\mathbf{u}_s\}_{s\in \mathcal{S}}} \quad &\tau +\frac{1}{\alpha}\left( \mu r^2+\frac{1}{S}\sum_{s\in \mathcal{S}}{\rho _s} \right)\\
	\mathrm{s}.\mathrm{t}.\quad
	&\mathbf{u}_s\in \widetilde{\mathcal{U} }(\boldsymbol{\theta },\mathbf{w}),\forall s\in \mathcal{S} ,\\
	&\rho _s\ge \delta \widetilde{\boldsymbol{\lambda }}_{s}^{\top}\mathbf{u}_s+\delta ^2\zeta_s-\tau ,\forall s\in \mathcal{S} ,\\
	&\left\| \mathbf{u}_s \right\| _2^2\le 4\mu \zeta_s,\forall s\in \mathcal{S} ,\\
	&\rho _s\ge 0,\zeta_s\ge 0,\forall s\in \mathcal{S} ,\\
	&\mu \ge 0.\\
\end{aligned}
}
\eeq
\end{proposition}
\begin{proof}
The detailed proof is given in Appendix~\ref{subsec:appendix_proof_prop1}.
\end{proof}

%Problem \eqref{eq:DRCVAR_problem_SOCP} is convex because, once $\boldsymbol{\theta}$ and $\mathbf{w}$ are specified, the aggregate-set membership constraint is linear, and $\left\| \mathbf{u}_s \right\| _2^2\le 4\mu \zeta_s$ with $\mu, \zeta_s\ge 0$ is a rotated second-order cone constraint. Thus, the forward utilization layer can be computed directly by off-the-shelf conic solvers, such as \texttt{GUROBI}. 

\subsubsection{Gradient Estimation in Backward Pass}
We then estimate the cost-informed gradients without differentiating through the building-level LPs. Let $\boldsymbol{\Theta}:=\rev{\left[\left(\boldsymbol{\theta}\right)^\top,\left(\mathbf{w}\right)^\top\right]^\top}\in \mathbb{R}^{3T+8}$ collect all learnable quantities, and write $J\left( \boldsymbol{\Theta} \right):=J\left( \boldsymbol{\theta},\mathbf{w} \right)$ for the corresponding forward-pass optimal value obtained by solving the building-level LPs
\eqref{eq:innerapp_problem_linear}, constructing the aggregate
set via \eqref{eq:innerapp_aggregate_affine}, and solving the
DR-CVaR SOCP \eqref{eq:DRCVAR_problem_SOCP}. Inspired by the perturbation-based approach in
\cite{berthet2020learning}, we optimize a smoothed objective. For a smoothing parameter $\sigma>0$, let $\boldsymbol{\Theta}^{\dagger}$ denote a Gaussian perturbation of $\boldsymbol{\Theta}$, i.e., $\boldsymbol{\Theta}^{\dagger}\sim \mathcal{N}\left( \boldsymbol{\Theta},\sigma ^2\mathbf{I}_{3T+8} \right)$. The smoothed objective is
\beq \label{eq:SS_objective} 
\eqfitu{
	J_{\sigma}\left( \boldsymbol{\Theta} \right):=\mathbb{E} _{\boldsymbol{\Theta }^{\dagger}}\left[ J\left( \boldsymbol{\Theta}^{\dagger} \right) \right] .
}
\eeq	

Equation \eqref{eq:SS_objective} defines randomized smoothing:
it averages the utilization cost over Gaussian perturbations
of the learnable parameters. Each perturbed cost
$J(\boldsymbol{\Theta}^{\dagger})$ is evaluated by solving
the building-level LPs and the downstream DR-CVaR SOCP
with the perturbed parameters. The gradient of this
smoothed objective can then be expressed using the
score-function identity as
  \beq \label{eq:SS_scorefunction_expectation} 
  \eqfitu{
\nabla _{\boldsymbol{\Theta}}J_{\sigma}\left( \boldsymbol{\Theta} \right)=\mathbb{E} _{\boldsymbol{\Theta }^{\dagger}}\left[ J\left( \boldsymbol{\Theta}^{\dagger} \right)\nabla _{\boldsymbol{\Theta}}\log p_{\boldsymbol{\Theta}}\left( \boldsymbol{\Theta}^{\dagger} \right) \right],
  }
  \eeq	
where $p_{\boldsymbol{\Theta}}$ denotes the Gaussian density
centered at $\boldsymbol{\Theta}$. Its score function with
respect to $\boldsymbol{\Theta}$ is
  \beq
\nabla_{\boldsymbol{\Theta}}
\log p_{\boldsymbol{\Theta}}
\left(\boldsymbol{\Theta}^{\dagger}\right)
=
\frac{\boldsymbol{\Theta}^{\dagger}-\boldsymbol{\Theta}}{\sigma^2}.
\eeq
Substituting this score function into
\eqref{eq:SS_scorefunction_expectation}, we approximate
the expectation using $N$ perturbed samples
$\{\boldsymbol{\Theta}_n^{\dagger}\}_{n=1}^N$,
each distributed as
$\mathcal{N}\left(\boldsymbol{\Theta},
\sigma^2\mathbf{I}_{3T+8}\right)$, to obtain
  \beq \label{eq:SS_scorefunction_sample} 
  \eqfitu{
\widehat{\nabla}J_{\sigma}\left( \boldsymbol{\Theta} \right)=\frac{1}{N\sigma ^2}\sum_{n=1}^N{J\left( \boldsymbol{\Theta}_n^{\dagger} \right)\left( \boldsymbol{\Theta}_n^{\dagger}-\boldsymbol{\Theta} \right)}.
  }
  \eeq
  This estimator requires only forward cost evaluations at
perturbed parameters, avoiding differentiation through
the building-level LPs.
  
\subsubsection{Tailored Solution Procedure}
Based on the forward SOCP reformulation and the backward gradient estimation, the overall learning procedure is summarized in Algorithm \ref{alg:learning_pipeline}.

\begin{algorithm}[ht]
	\caption{Cost-informed learning procedure.}
	\label{alg:learning_pipeline}
	\begin{algorithmic}[1]
		\State \textbf{Input:} individual sets $\left\{ \mathcal{U}_b \right\}_{b\in \mathcal{B}}$, price scenarios $\left\{ \widetilde{\boldsymbol{\lambda}}_s \right\}_{s\in \mathcal{S}}$, initial parameters $\boldsymbol{\Theta}_0=\rev{\left[\left(\boldsymbol{\theta}_0\right)^\top,\left(\mathbf{w}_0\right)^\top\right]^\top}$, smoothing variance $\sigma^2$, sample size $N$, iteration number $K$, step size $\nu$.
		\For{$k=0,1,\ldots,K-1$}
		\For{$n=1,\ldots,N$}
		\State Sample $\boldsymbol{\Theta}_{k,n}^{\dagger}\sim \mathcal{N}\left( \boldsymbol{\Theta}_k,\sigma ^2\mathbf{I}_{3T+8} \right)$.
		\State Recover $\boldsymbol{\theta}_{k,n}^{\dagger}$ and $\mathbf{w}_{k,n}^{\dagger}$ from $\boldsymbol{\Theta}_{k,n}^{\dagger}$.
		\State Solve the building-level LPs \eqref{eq:innerapp_problem_linear} for all $b\in \mathcal{B}$.
		\State Construct $\widetilde{\mathcal{U}}\left( \boldsymbol{\theta}_{k,n}^{\dagger},\mathbf{w}_{k,n}^{\dagger} \right)$ using \eqref{eq:innerapp_aggregate_affine}.
		\State Solve the DR-CVaR SOCP \eqref{eq:DRCVAR_problem_SOCP} and record $J_{k,n}$.
		\EndFor
		\State Estimate $\widehat{\nabla}J_{\sigma}\left( \boldsymbol{\Theta}_k \right)$ using \eqref{eq:SS_scorefunction_sample}.
		\State Update $\boldsymbol{\Theta}_{k+1}\leftarrow \boldsymbol{\Theta}_k-\nu \widehat{\nabla}J_{\sigma}\left( \boldsymbol{\Theta}_k \right)$.
		\EndFor
		\State \textbf{Output:} learned aggregate set $\widetilde{\mathcal{U}}\left( \boldsymbol{\theta}_K,\mathbf{w}_K \right)$.
	\end{algorithmic}
\end{algorithm}

\begin{remark}[Parallel Computation]\label{remark:parallel_computation}
	The proposed procedure has a parallel computational structure. For each perturbed parameter vector, the building-level LPs in \eqref{eq:innerapp_problem_linear}, corresponding to Step~6 of Algorithm~\ref{alg:learning_pipeline}, are independent across buildings and can be solved in parallel. Once the aggregate set is constructed, the DR-CVaR SOCP in \eqref{eq:DRCVAR_problem_SOCP} depends only on the aggregate affine surrogate, rather than on individual building models. Therefore, the computational burden of the utilization layer does not grow with the number of buildings. In addition, the Monte Carlo evaluations used for gradient estimation are independent across perturbation samples, providing another level of parallelism in Step~10 of Algorithm~\ref{alg:learning_pipeline}. Hence, parallel implementation can reduce the overall computation time, supporting scalable learning for large HVAC fleets.
\end{remark}

\subsection{Implementation Workflow}\label{subsec:framework_workflow}
The proposed framework can be implemented in three phases, as illustrated in Fig. \ref{fig:Workflow}.
\begin{itemize}
	\item \textbf{Phase 1: Model preparation.} The aggregator collects individual HVAC flexibility models $\left\{ \mathcal{U}_b \right\}_{b\in\mathcal{B}}$ and historical or forecast price trajectories $\left\{ \widetilde{\boldsymbol{\lambda}}_s \right\}_{s\in\mathcal{S}}$. These price trajectories define the nominal distribution $P_0$, while $\alpha$ and $r$ specify the tail-risk preference and the confidence in the nominal price distribution.
	
	\item \textbf{Phase 2: Cost-informed learning.} Using the prepared HVAC models and price scenarios, the aggregator runs Algorithm \ref{alg:learning_pipeline}. This phase jointly updates the storage-surrogate parameters and inner-approximation objective through downstream DR-CVaR feedback, and returns the learned aggregate set $\widetilde{\mathcal{U}}\left( \boldsymbol{\theta}_K,\mathbf{w}_K \right)$.
	
	\item \textbf{Phase 3: Dispatch and disaggregation.} The learned aggregate set is submitted to the system operator,
	which determines an aggregate schedule $\mathbf{u}^{\star}$
	within the submitted set and the associated surrogate charge/discharge profile $\rev{\left[\left(\mathbf{u}^{\mathrm{C}\star}\right)^\top,\left(\mathbf{u}^{\mathrm{D}\star}\right)^\top\right]^\top}\in\mathcal{U}_0\left( \boldsymbol{\theta}_K \right)$. In the numerical experiments, this
	utilization stage is evaluated using the realized-price proxy
	in \eqref{eq:utilize_realizedprice}. Since the learned set is an inner approximation, the aggregate schedule can be disaggregated into feasible building-level schedules by the affine maps obtained during learning:
	  \beq \label{eq:disaggregation} 
	\eqfitu{
	\mathbf{u}_{b}^{\star}=\boldsymbol{\gamma }_{b,K}+\boldsymbol{\Phi}\boldsymbol{\Gamma }_{b,K}\rev{\left[\left(\mathbf{u}^{\mathrm{C}\star}\right)^\top,\left(\mathbf{u}^{\mathrm{D}\star}\right)^\top\right]^\top},\quad \forall b\in \mathcal{B}.
	}
	\eeq	
	These schedules satisfy $\mathbf{u}_{b}^{\star}\in\mathcal{U}_b$ and $\mathbf{u}^{\star}=\sum_{b\in\mathcal{B}}{\mathbf{u}_{b}^{\star}}$.
\end{itemize}

\begin{figure}[!t]
	\centering
	\includegraphics[width=\linewidth]{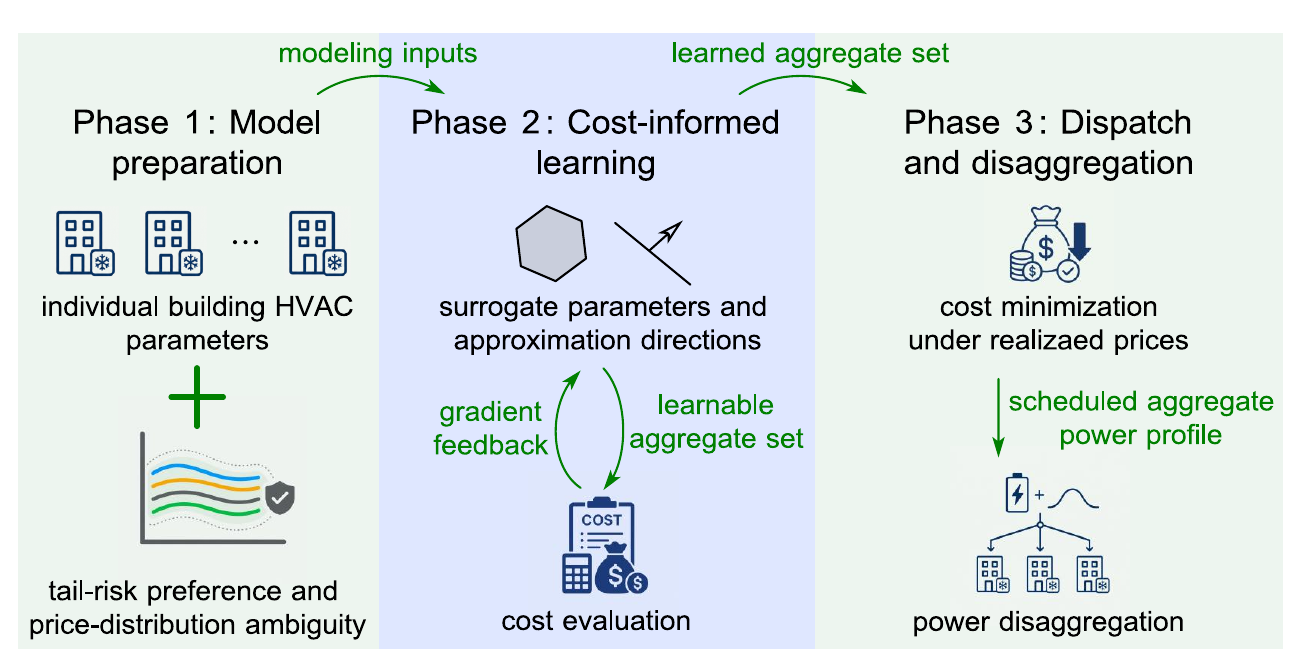}
	\caption{Implementation workflow of the proposed cost-informed framework}
	\label{fig:Workflow}
\end{figure}

%% file: IV_Case.tex
In this section, we present case studies to demonstrate the economic benefit of the proposed cost-informed learning framework and the effectiveness of its risk- and ambiguity-aware modeling.

\subsection{Simulation Setup}\label{subsec:case_setup}
\subsubsection{Parameter Settings}
We consider a 24-hour horizon with an hourly time interval ($\delta=1$) and an HVAC fleet consisting of 30 buildings ($B=30$). The HVAC parameter intervals are summarized in Table~\ref{tab:parameter_HVAC} where the parameters of each building are sampled uniformly from these intervals. The electricity prices are taken from the 2025 NYISO New York City Zone J day-ahead price data \cite{NYCPrice_2025}. We use March 2025 as the training period, April 2025 as an in-distribution testing period, and July 2025 as an out-of-distribution testing period. The corresponding price trajectories are shown in Appendix~\ref{subsec:appendix_price}. The April test set evaluates performance when the testing distribution is close to the training distribution, whereas the July test set provides a more demanding assessment under distribution shift and high-price tail events.

\begin{table}[!t]
	\footnotesize
	\caption{HVAC Parameters}
	\begin{center}
		\begin{tabular}{ ccc } 
			\hline
			Parameter& Description & Value \\ 
			\hline
			$r_{b}^{\mathrm{I}}$ & Thermal resistance of air & $[3, 5]$ $^\circ$C/kW \\
			$r_{b}^{\mathrm{M}}$ & Thermal resistance of mass & $[5, 7]$ $^\circ$C/kW \\
			$c_{b}^{\mathrm{I}}$ & Thermal capacity of air & $[5, 7]$ kWh/$^\circ$C \\
			$c_{b}^{\mathrm{M}}$ & Thermal capacity of mass & $[6, 8]$ kWh/$^\circ$C\\
			$\eta _b$ & Coefficient of performance  & $[2,3]$ \\
			$\alpha _b$ &  Fractional coefficient of power & $[0.03, 0.06]$  \\
			$x_{b}^{\mathrm{set}}$ & Temperature setpoint  & $[19, 23]$  $^\circ$C \\
			$x_{b}^{\mathrm{tol}}$ & Comfort tolerance band  & $[1.5,2.5]$  $^\circ$C \\
			$\underline{u}_b$ & Minimum HVAC power  & $[0.3,0.7]$ kW \\
			$\overline{u}_b$ & Maximum HVAC power  & $[6,7]$ kW \\
			\hline
		\end{tabular}
		\label{tab:parameter_HVAC}
	\end{center}
\end{table}

\subsubsection{Test Platform}
All numerical experiments are conducted on a laptop with a 2.20 GHz CPU and 32 GB of RAM. The optimization problems are solved using \texttt{GUROBI}. We train the proposed method for 50 iterations with the step size $\nu=0.001$ in Algorithm~\ref{alg:learning_pipeline}. Following standard practice \cite{berthet2020learning}, we use $N=2$ Gaussian perturbations per iteration.

\subsubsection{Benchmarks}
We compare the following methods:
\begin{itemize}
	\item \textsf{Oracle}: a non-aggregation benchmark that directly solves the original utilization problem with access to all individual HVAC flexibility sets and realized prices. It serves as a best-case performance reference rather than a deployable aggregation method.
	\item \textsf{VMA}: the volume-maximization aggregation method used in \cite{al2024efficient,liu2026coupling}, which constructs a maximum-volume inner approximation without cost-informed feedback.
	\item \textsf{Proposed}: the proposed cost-informed learning method, which updates the surrogate parameters and inner-approximation objective with cost-informed feedback.
\end{itemize}

For each price scenario, we first compute the utilization cost by solving \eqref{eq:utilize_realizedprice} over the aggregate set produced by each method. The \emph{cost gap} is then defined as the percentage increase of this cost relative to \textsf{Oracle}, with lower values indicating better economic performance. Repeating this calculation over all days in a training or testing period gives the cost-gap distributions reported below.

\subsection{Cost-Informed Learning Results}\label{subsec:case_learning_result}
We first evaluate cost-informed learning using March 2025 for training and April 2025 for testing. Since the March and April price distributions are close and contain no severe tail-risk testing events, we set $\alpha=1$ and $r=0.01$, corresponding to a risk-neutral DR-expected setting with a small ambiguity radius. The simulation results are shown in Fig.~\ref{fig:ConvergenceAnalysis}--Fig.~\ref{fig:DispatchResults_Temperature}.

\begin{figure}[!t]
	\centering
	\includegraphics[width=\linewidth]{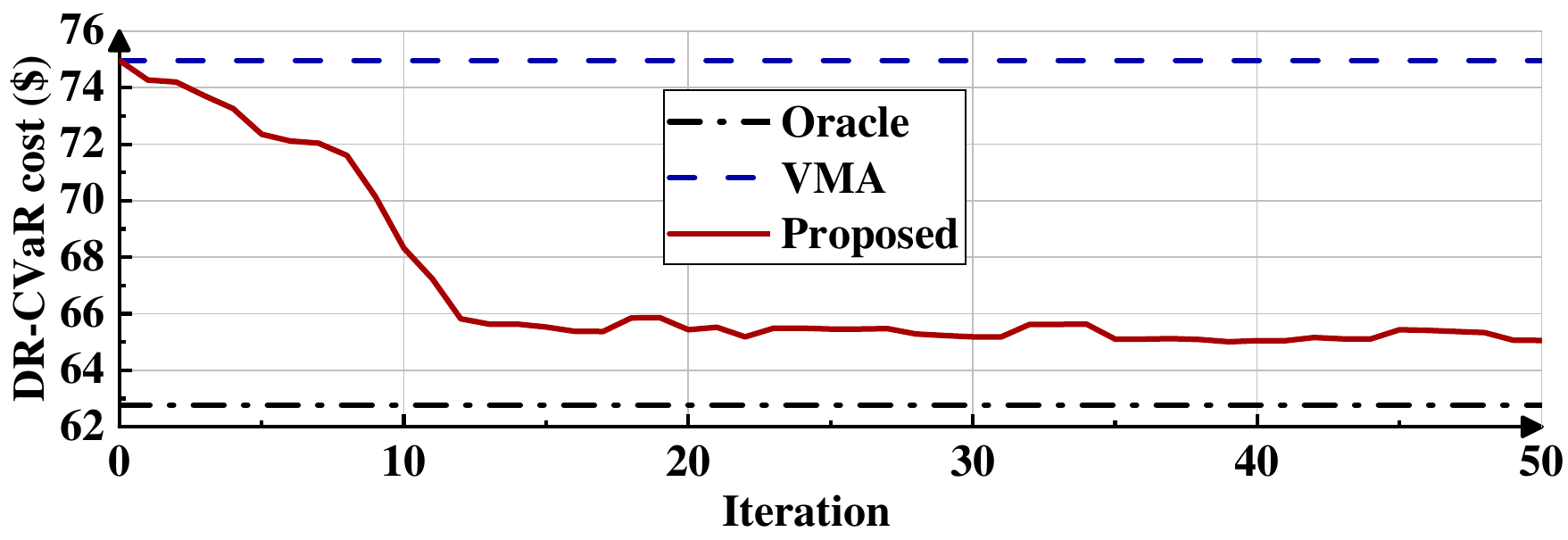}
	\caption{Convergence trajectories of the DR-CVaR utilization cost.}
	\label{fig:ConvergenceAnalysis}
\end{figure}

\begin{figure}[!t]
	\centering
	\includegraphics[width=\linewidth]{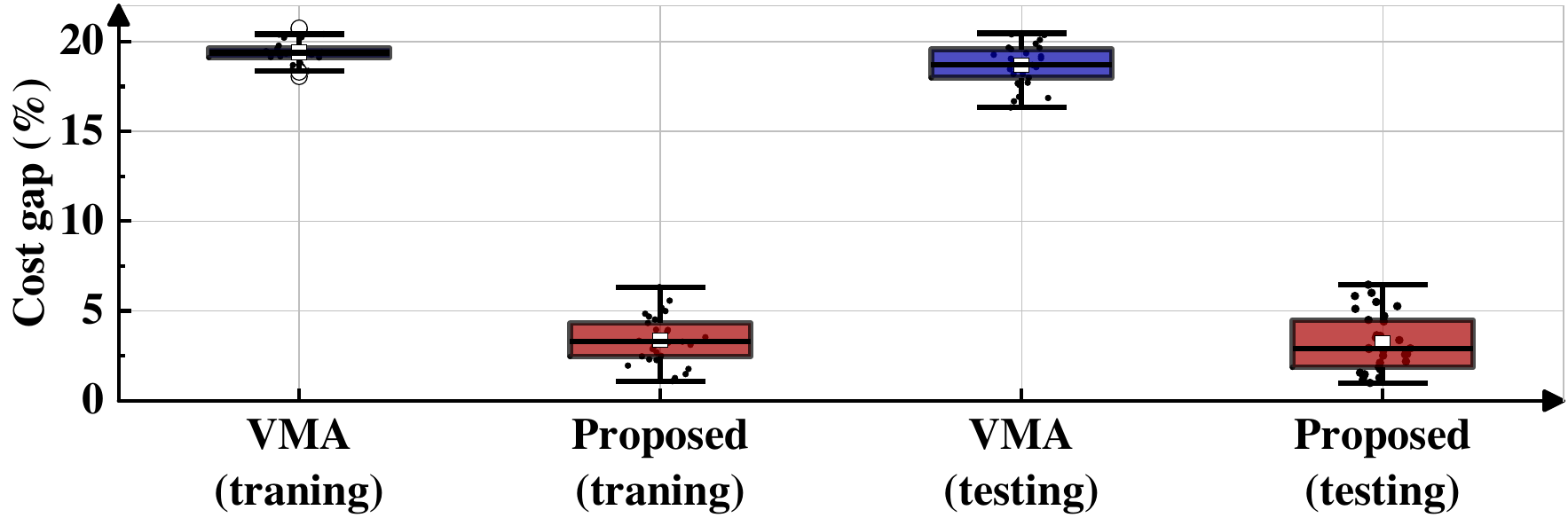}
	\caption{Training and testing cost-gap distributions for different methods.}
	\label{fig:CostResults}
\end{figure}

\begin{figure}[!t]
	\centering
	\includegraphics[width=\linewidth]{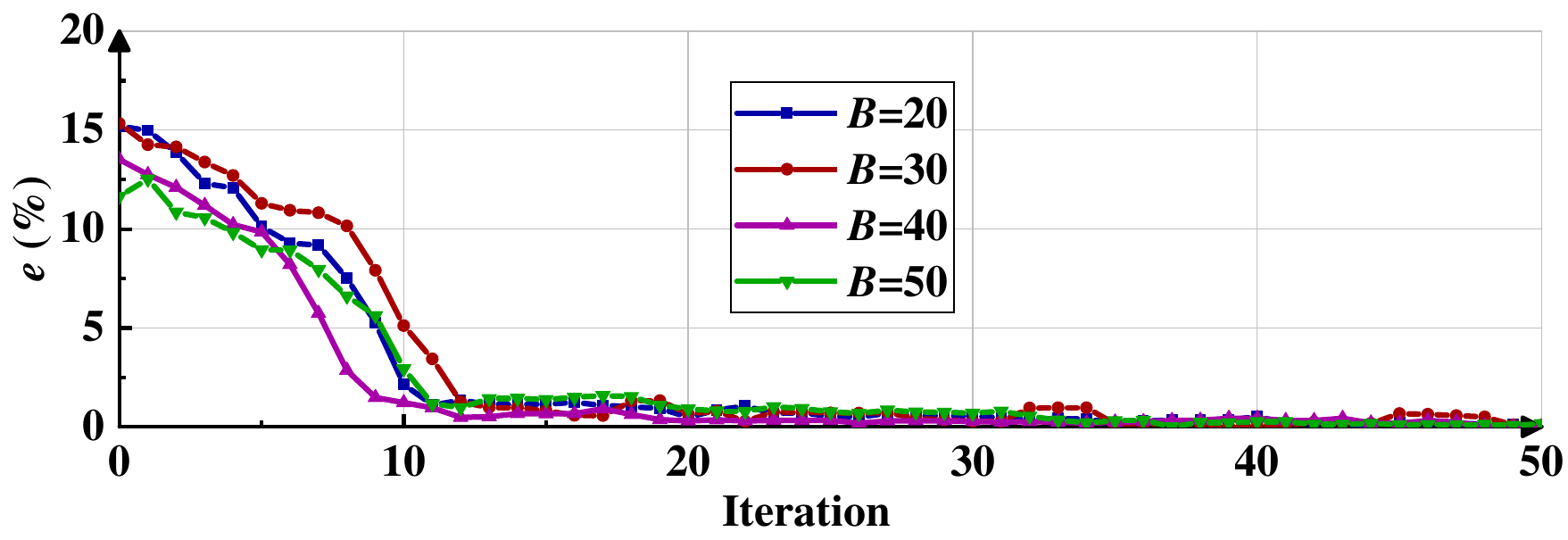}
	\caption{Convergence of normalized suboptimality for different fleet sizes.}
	\label{fig:ComputationalEfficiency_profile}
\end{figure}

\begin{figure}[!t]
	\centering
	\includegraphics[width=\linewidth]{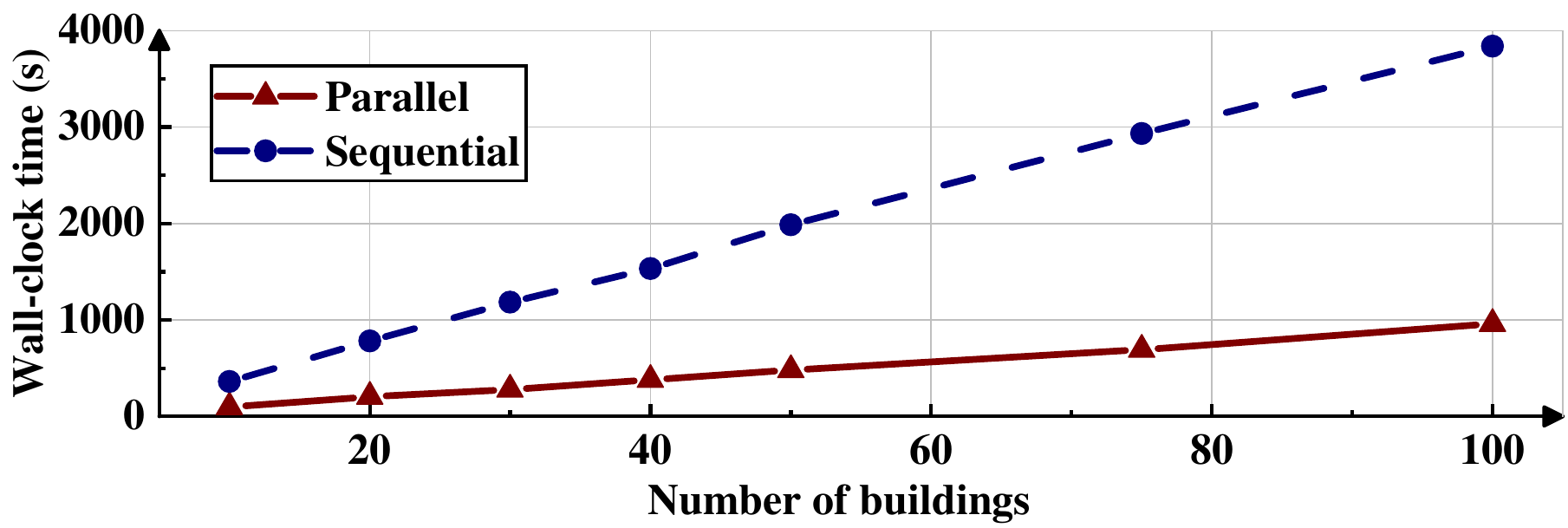}
	\caption{Total wall-clock time of the learning procedure for different fleet sizes under sequential and parallel execution.
	}
	\label{fig:ComputationalEfficiency_CPUtime}
\end{figure}

\begin{figure}[!t]
	\centering
	\includegraphics[width=\linewidth]{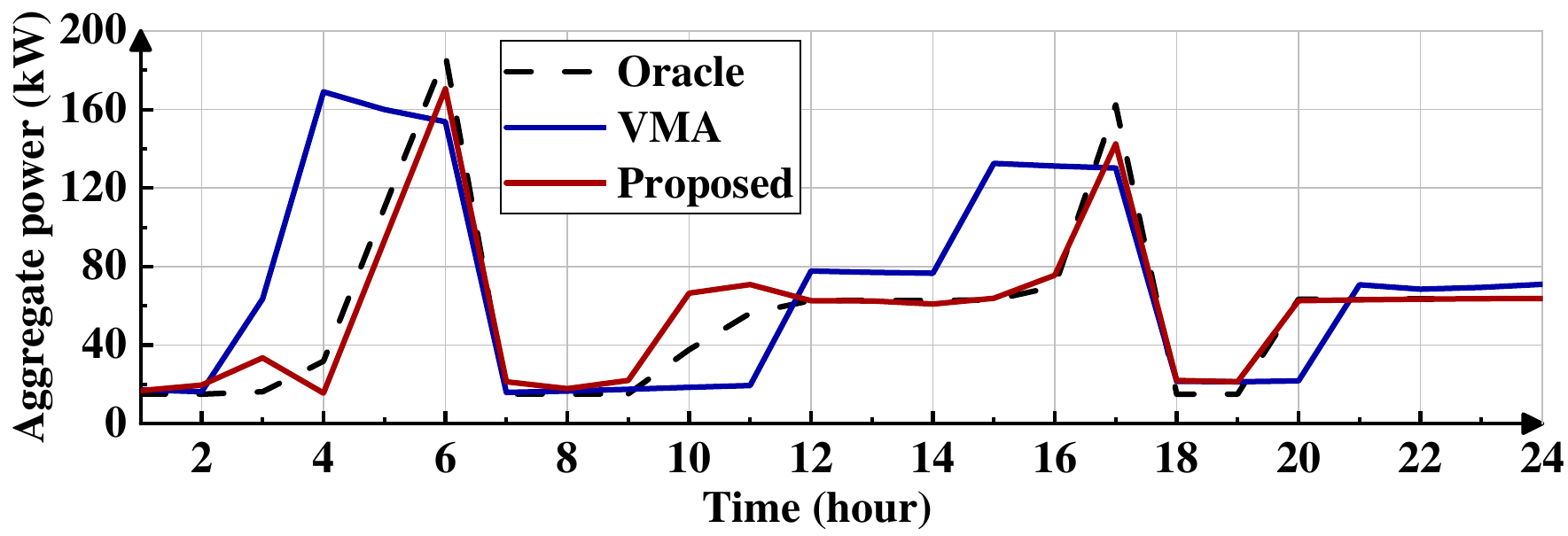}
	\caption{Aggregate power profiles on March 5, 2025.}
	\label{fig:DispatchResults_Power}
\end{figure}

\begin{figure}[!t]
	\centering
	\includegraphics[width=\linewidth]{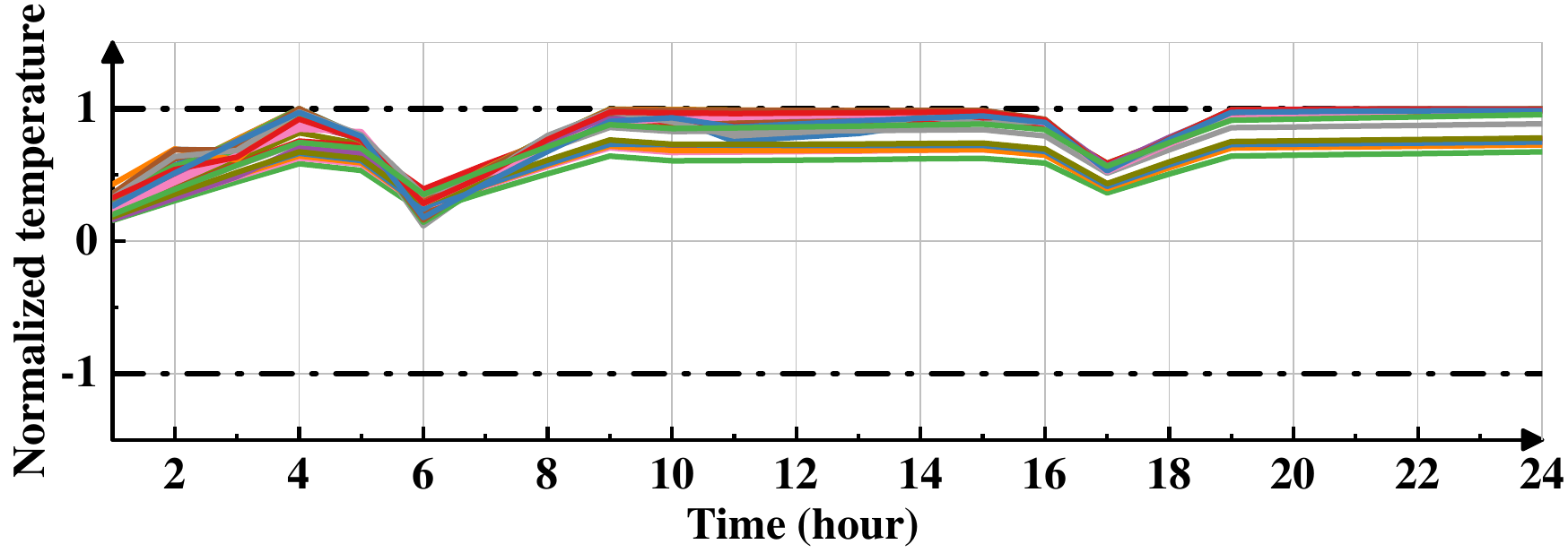}
	\caption{Disaggregated normalized building temperatures obtained by \eqref{eq:disaggregation} on March 5, 2025. The black dash-dotted lines indicate user comfort bounds, while the colored profiles show the normalized temperature trajectories of individual buildings.}
	\label{fig:DispatchResults_Temperature}
\end{figure}

\subsubsection{Convergence Analysis}
Fig.~\ref{fig:ConvergenceAnalysis} shows the learning trajectory of \textsf{Proposed}. Since the procedure is initialized from the volume-maximization aggregation solution, as discussed in Remark~\ref{remark:compare_bitar}, its initial utilization cost is the same as that of \textsf{VMA}. After the storage-surrogate parameters and inner-approximation objective are updated, the cost decreases and moves toward the \textsf{Oracle} benchmark, while \textsf{VMA} remains fixed as the initialization reference. For this training experiment, the learned aggregate set therefore attains a lower utilization cost than the initial volume-oriented set.

\subsubsection{Cost Gap}
Fig.~\ref{fig:CostResults} evaluates the learned aggregate set over all March training scenarios and April testing scenarios. In both distributions, \textsf{Proposed} achieves lower cost gaps than \textsf{VMA}, showing that the improvement observed during training also transfers to the in-distribution test period. This comparison illustrates that the larger-volume inner approximation is not necessarily more valuable for downstream dispatch. Its value also depends on whether it preserves flexibility in economically active directions.

\subsubsection{Computational Efficiency}
We next evaluate computational efficiency using Fig.~\ref{fig:ComputationalEfficiency_profile} and Fig.~\ref{fig:ComputationalEfficiency_CPUtime}. In Fig.~\ref{fig:ComputationalEfficiency_profile}, the normalized suboptimality is defined as
\beq
e_k:=\frac{J\left( \boldsymbol{\Theta}_k \right)-J\left( \boldsymbol{\Theta}_K \right)}
{J\left( \boldsymbol{\Theta}_K \right)}\times 100\% ,
\eeq
where $\boldsymbol{\Theta}_k$ denotes the learnable parameters at iteration $k$ and $K$ is the final iteration. For $B\in\{20,30,40,50\}$, the learning trajectories become stable after approximately 30 iterations. This behavior is consistent with the fact that the number of learnable storage-surrogate parameters and inner-approximation objective does not grow with the number of buildings, although the building-level LP workload does.

Fig.~\ref{fig:ComputationalEfficiency_CPUtime} further shows the resulting total iteration computation times as a function of the fleet size $B$. For $B = 100$, computation takes $3838.5$~s sequentially and $960.5$~s with 12 parallel workers. Under the fixed 12-worker configuration, parallel computation times increase approximately linearly with fleet size. These results demonstrate the computational benefits of the building-separable architecture for large-fleet aggregation. Additional processor cores or distributed workers could further reduce computation time, ensuring that the proposed learning procedure scales efficiently with the number of buildings.

\subsubsection{Dispatch Results}
Fig.~\ref{fig:DispatchResults_Power} compares the aggregate power profiles on March 5, 2025. The profile produced by \textsf{Proposed} is much closer to \textsf{Oracle} than the profile produced by \textsf{VMA}. The difference is most visible at cost-sensitive operating periods. Around period~6, the economic action is to increase cooling power in advance of a price peak, thereby precooling the buildings. Around period~8, the aggregator then reduces power consumption and uses the stored thermal flexibility. A similar pattern occurs over periods~17--19: the dispatch increases power in period~17 and reduces power during the subsequent high-price periods.

However, \textsf{VMA} does not explicitly distinguish these cost-sensitive directions from less relevant directions. As a result, it can preserve aggregate volume while still losing flexibility that is valuable for the downstream utilization problem. \textsf{Proposed} instead reallocates representation capacity toward the regions activated by the cost-minimizing dispatch, which explains its closer agreement with \textsf{Oracle}.

Fig.~\ref{fig:DispatchResults_Temperature} illustartes the closed-form disaggregation strategy in \eqref{eq:disaggregation}. The normalized temperature is defined as $\left( x_{b,t}^{\mathrm{I}}-x_{b}^{\mathrm{set}} \right)/x_{b}^{\mathrm{tol}}$. All normalized building temperatures remain within $[-1,1]$, showing that the aggregate schedule obtained from the learned inner approximation can be disaggregated into feasible building-level schedules that satisfy the comfort constraints.

In summary, the above results demonstrate that the proposed cost-informed learning method improves downstream economic performance while preserving computational tractability and disaggregation feasibility.

\subsection{Risk and Ambiguity Sensitivity Analysis}\label{subsec:case_unseen_price}

\begin{figure}[!t]
	\centering
	\includegraphics[width=\linewidth]{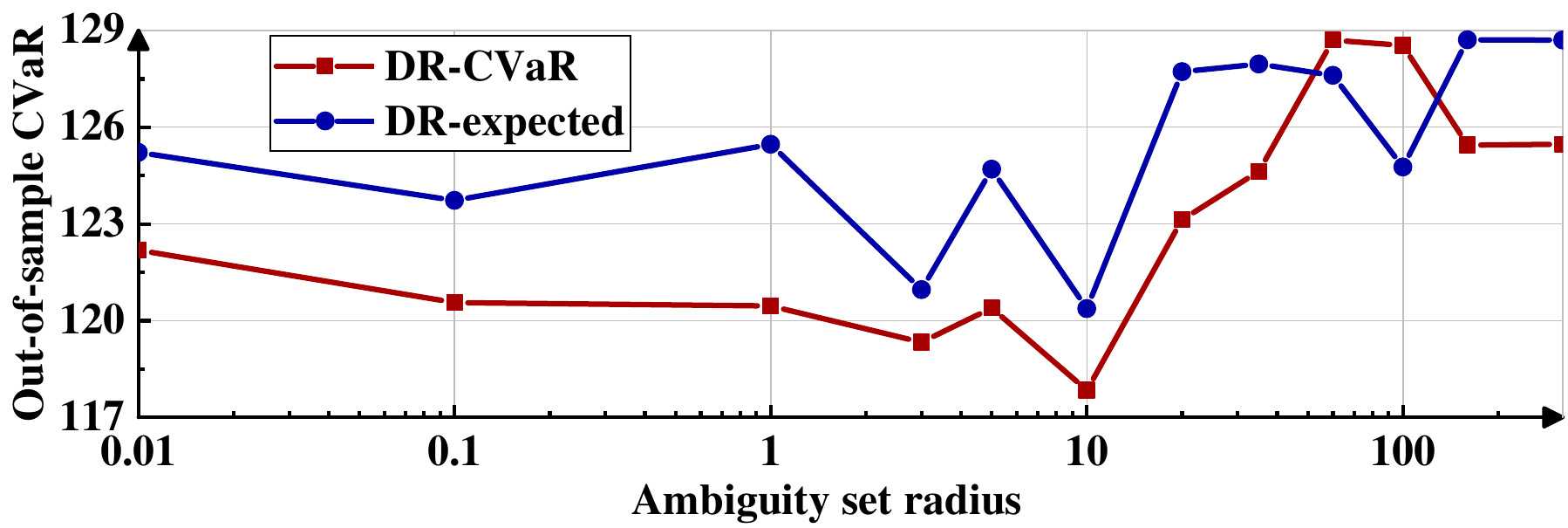}
	\caption{Out-of-sample CVaR of testing costs versus ambiguity set radius $r$.}
	\label{fig:Out-of-SampleCVaR}
\end{figure}

\begin{figure}[!t]
	\centering
	\includegraphics[width=\linewidth]{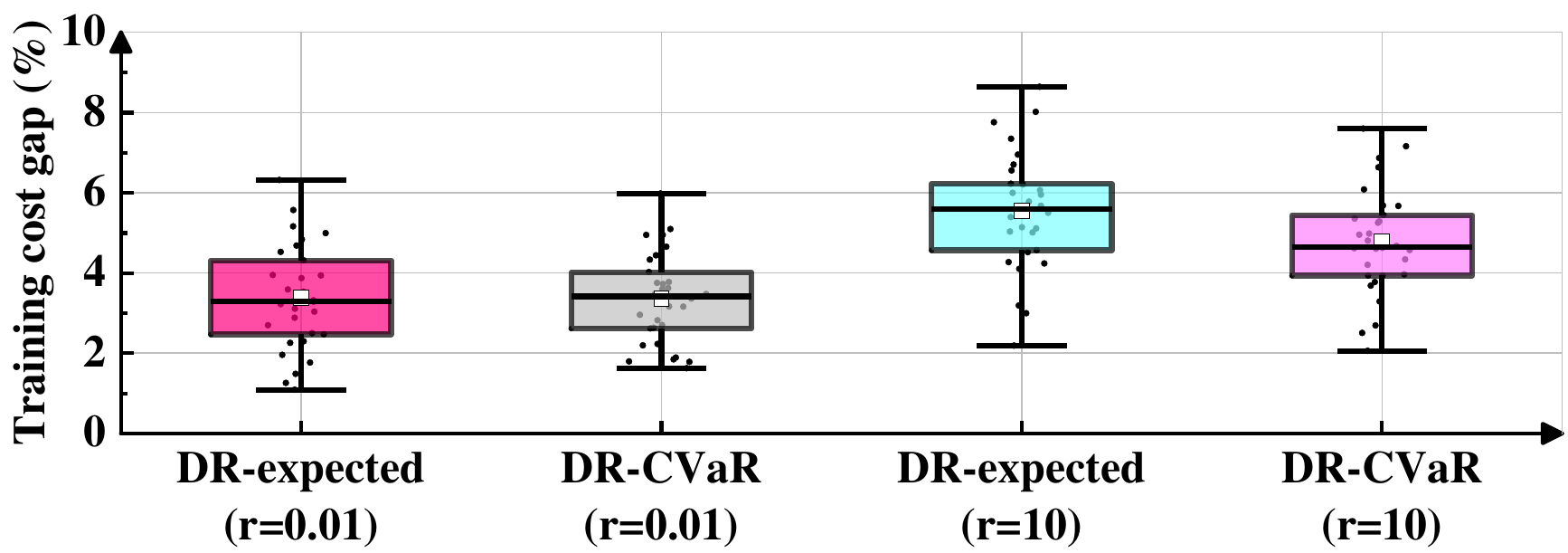}
	\caption{Training cost-gap distributions under different risk and ambiguity settings.}
	\label{fig:GeneralizationPerformance_train}
\end{figure}

\begin{figure}[!t]
	\centering
	\includegraphics[width=\linewidth]{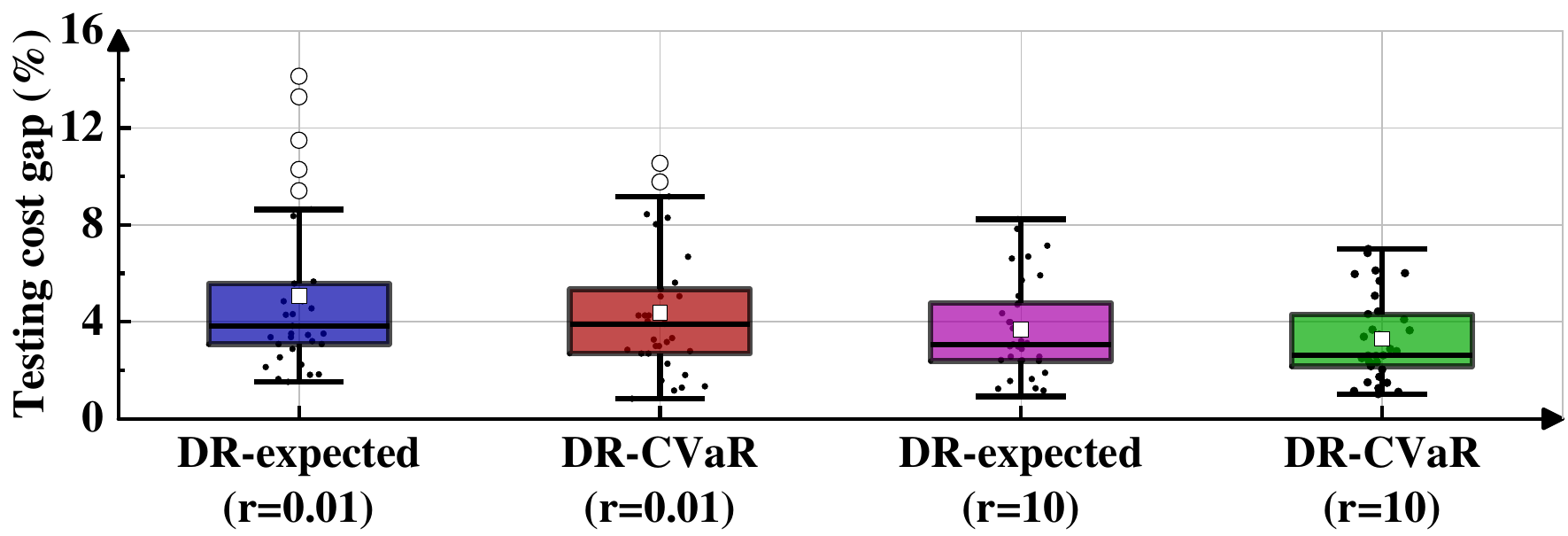}
	\caption{Testing cost-gap distributions under different risk and ambiguity settings.}
	\label{fig:GeneralizationPerformance_test}
\end{figure}

We further evaluate the risk- and ambiguity-aware modeling component by using March 2025 data for training and July 2025 data for testing. The July data contain high-price tail scenarios and exhibit a stronger distribution shift from the March training distribution, making them suitable for assessing out-of-sample generalization. We compare the DR-expected objective with $\alpha=1$ and the DR-CVaR objective with $\alpha=0.1$. The simulation results are shown in Fig.~\ref{fig:Out-of-SampleCVaR}--Fig.~\ref{fig:GeneralizationPerformance_test}.

\subsubsection{Out-of-Sample CVaR}

Fig.~\ref{fig:Out-of-SampleCVaR} shows the out-of-sample CVaR of July testing costs under different ambiguity radii. Both methods are evaluated using the empirical CVaR at the
same tail probability level $\alpha_{\mathrm{eval}}=0.1$
on the July testing scenarios. When $r$ is close to zero, the ambiguity set is close to the nominal March training distribution. Increasing $r$ introduces distributional robustness. However, the testing CVaR does not decrease monotonically as $r$ grows. A moderate radius, especially around $r=10$, yields the lowest testing CVaR among the tested settings. This is a post hoc sensitivity observation from the July results rather than a value selected using an independent calibration set. In the small-to-moderate radius range, the DR-CVaR formulation generally achieves a lower out-of-sample CVaR than the DR-expected formulation. This indicates that incorporating price-distribution ambiguity and tail-risk preferences into aggregate-set learning can improve generalization under the July price shift while reducing tail risk. When $r$ becomes too large, the learned aggregate set may become overly conservative or mismatched with the testing distribution, and the out-of-sample CVaR no longer improves.

\subsubsection{Generalization Performance}

Fig.~\ref{fig:GeneralizationPerformance_train} and Fig.~\ref{fig:GeneralizationPerformance_test} further compare the training and testing cost-gap distributions. In the March training scenarios, small-radius settings fit the nominal training distribution well, whereas $r=10$ introduces additional robustness and may increase the in-sample gap. This moderate loss of in-sample performance reflects the intended tradeoff for improving robustness under distribution shift.

The July testing results demonstrate the benefit of this tradeoff. With $r=0.01$, both formulations exhibit larger testing tails, and the tail gap of the DR-expected formulation reaches about 15\%. Increasing the radius to $r=10$ reduces the testing tail gaps, with the DR-CVaR formulation producing the most compact tail and keeping the gap below about 8\%. This reduction in testing tail risk indicates that a suitable ambiguity radius among the tested settings improves robustness to price-distribution shifts, while the DR-CVaR formulation helps the learned aggregate set avoid high-cost tail outcomes.

Overall, these results demonstrate that the proposed risk- and ambiguity-aware modeling component improves the generalization of the learned aggregate set to unseen price scenarios while reducing tail risk.

%% file: Appendix.tex
\section{Proof of Proposition \ref{prop:DRCVAR_problem_SOCP}}\label{subsec:appendix_proof_prop1}
Fix $\boldsymbol{\theta}$ and $\mathbf{w}$, and write $\widetilde{\mathcal{U}}:=\widetilde{\mathcal{U}}\left( \boldsymbol{\theta},\mathbf{w} \right)$. Since $\widetilde{\mathcal{U}}$ is nonempty, compact, and convex, $V\left( \boldsymbol{\lambda};\widetilde{\mathcal{U}} \right)$ in \eqref{eq:utilize_realizedprice} is finite, continuous, and has at most linear growth.

The proof follows the Wasserstein DR-CVaR duality in Lemma \ref{lemma:DRCVAR_dual} and finite-support epigraph reformulation in \cite[Theorem 3]{taha2026wasserstein}, specialized to the utilization value $V\left( \boldsymbol{\lambda};\widetilde{\mathcal{U}} \right)$. By Lemma \ref{lemma:DRCVAR_dual}, the worst-case CVaR problem \eqref{eq:DRCVAR_problem_original} admits the Wasserstein-CVaR dual reformulation \eqref{eq:DRCVaR_dual_formulation}. Since $P_0$ is uniform over the finite support $\{\widetilde{\boldsymbol{\lambda}}_s\}_{s\in\mathcal{S}}$, the expectation in \eqref{eq:DRCVaR_dual_formulation} becomes the sample average
\beq
\eqfitu{
J\left( \boldsymbol{\theta},\mathbf{w} \right)=\min_{\tau \in \mathbb{R},\,\mu \ge 0}\left\{ \tau +\frac{1}{\alpha}\left( \mu r^2+\frac{1}{S}\sum_{s\in\mathcal{S}}{\phi\left( \tau,\mu,\widetilde{\boldsymbol{\lambda}}_s \right)} \right) \right\},
}
\eeq
where $\phi\left( \tau,\mu,\widetilde{\boldsymbol{\lambda}}_s \right)$ denotes the function in \eqref{eq:DRCVaR_dual_definition} evaluated at the finite-support sample $\widetilde{\boldsymbol{\lambda}}_s$.

For $\mu>0$, using $\left( x \right)_+=\max \left\{ x,0 \right\}$ gives
\beq
\eqfitu{
\begin{aligned}
\phi\left( \tau,\mu,\widetilde{\boldsymbol{\lambda}}_s \right)
&=\max\left\{
\max_{\boldsymbol{\lambda}\in\mathbb{R}^T}\left( V\left( \boldsymbol{\lambda};\widetilde{\mathcal{U}} \right)-\tau-\mu \left\| \boldsymbol{\lambda}-\widetilde{\boldsymbol{\lambda}}_s \right\|_2^2 \right),
\max_{\boldsymbol{\lambda}\in\mathbb{R}^T}\left( -\mu \left\| \boldsymbol{\lambda}-\widetilde{\boldsymbol{\lambda}}_s \right\|_2^2 \right)
\right\} \\
&=\max\left\{
\max_{\boldsymbol{\lambda}\in\mathbb{R}^T}\left( V\left( \boldsymbol{\lambda};\widetilde{\mathcal{U}} \right)-\tau-\mu \left\| \boldsymbol{\lambda}-\widetilde{\boldsymbol{\lambda}}_s \right\|_2^2 \right),
0
\right\}.
\end{aligned}
}
\eeq

The second maximization $\max_{\boldsymbol{\lambda }\in \mathbb{R} ^T} \left( -\mu \left\| \boldsymbol{\lambda }-\widetilde{\boldsymbol{\lambda }}_s \right\| _{2}^{2} \right) $ equals $0$, attained at $\boldsymbol{\lambda}=\widetilde{\boldsymbol{\lambda}}_s$. From \eqref{eq:utilize_realizedprice}, we have
\beq
\eqfitu{
V\left( \boldsymbol{\lambda};\widetilde{\mathcal{U}} \right)=\min_{\mathbf{u}\in\widetilde{\mathcal{U}}}\delta \boldsymbol{\lambda}^{\top}\mathbf{u}.
}
\eeq
After substituting \eqref{eq:utilize_realizedprice} into the remaining maximization term in the preceding display, the minimax interchange is valid because $\widetilde{\mathcal{U}}$ is compact and convex, the objective is affine in $\mathbf{u}$ and strictly concave and coercive in $\boldsymbol{\lambda}$ for $\mu>0$. Hence,
\beq
\eqfitu{
\begin{aligned}
\max_{\boldsymbol{\lambda}\in\mathbb{R}^T}\min_{\mathbf{u}\in\widetilde{\mathcal{U}}}\left\{ \delta \boldsymbol{\lambda}^{\top}\mathbf{u}-\tau-\mu \left\| \boldsymbol{\lambda}-\widetilde{\boldsymbol{\lambda}}_s \right\|_2^2 \right\}
=\min_{\mathbf{u}\in\widetilde{\mathcal{U}}}\max_{\boldsymbol{\lambda}\in\mathbb{R}^T}\left\{ \delta \boldsymbol{\lambda}^{\top}\mathbf{u}-\tau-\mu \left\| \boldsymbol{\lambda}-\widetilde{\boldsymbol{\lambda}}_s \right\|_2^2 \right\}.
\end{aligned}
}
\eeq
For fixed $\mathbf{u}$, the first-order condition of the inner maximization is
\beq
\eqfitu{
\delta \mathbf{u}-2\mu\left( \boldsymbol{\lambda}-\widetilde{\boldsymbol{\lambda}}_s \right)=0,
}
\eeq
so $\boldsymbol{\lambda}^{\star}=\widetilde{\boldsymbol{\lambda}}_s+\delta \mathbf{u}/(2\mu)$. Substituting this optimizer yields
\beq
\eqfitu{
\phi\left( \tau,\mu,\widetilde{\boldsymbol{\lambda}}_s \right)=\left( \min_{\mathbf{u}_s\in\widetilde{\mathcal{U}}}\left\{ \delta \widetilde{\boldsymbol{\lambda}}_s^{\top}\mathbf{u}_s+\frac{\delta ^2\left\| \mathbf{u}_s \right\|_2^2}{4\mu}-\tau \right\} \right)_+ .
}
\eeq
Since $\widetilde{\mathcal{U}}$ is nonempty and compact, the minimum is attained and
\beq
\eqfitu{
\left(\min_{\mathbf{u}_s\in\widetilde{\mathcal{U}}}
\left\{\delta \widetilde{\boldsymbol{\lambda}}_s^\top\mathbf{u}_s+
\frac{\delta^2\left\|\mathbf{u}_s\right\|_2^2}{4\mu}-\tau\right\}\right)_+
=\min_{\mathbf{u}_s\in\widetilde{\mathcal{U}}}
\max\left\{\delta \widetilde{\boldsymbol{\lambda}}_s^\top\mathbf{u}_s+
\frac{\delta^2\left\|\mathbf{u}_s\right\|_2^2}{4\mu}-\tau,0\right\}.
}
\eeq
Introduce $\rho_s$ to represent the positive part and $\zeta_s$ to represent $\left\|\mathbf{u}_s\right\|_2^2/(4\mu)$ through a rotated second-order cone constraint. The corresponding epigraph and rotated second-order cone constraints are
\beq
\eqfitu{
\rho_s\ge \delta \widetilde{\boldsymbol{\lambda}}_s^{\top}\mathbf{u}_s+\delta ^2\zeta_s-\tau,\quad
\rho_s\ge 0,\quad
\left\| \mathbf{u}_s \right\|_2^2\le 4\mu \zeta_s,\quad
\zeta_s\ge 0.
}
\eeq
Since $\rho_s$ enters the sample-average dual objective with a positive coefficient and is minimized, it equals the positive part at optimum. Moreover, $\zeta_s$ can be chosen as the smallest value permitted by the rotated second-order cone constraint. The preceding derivation is for $\mu>0$; the conic representation with $\mu\ge0$ follows by closure. Substituting these epigraph and rotated second-order cone constraints into \eqref{eq:DRCVaR_dual_formulation} yields the objective and constraints of \eqref{eq:DRCVAR_problem_SOCP}. This proves the SOCP reformulation. \QED

\section{Price Trajectories}\label{subsec:appendix_price}
See Fig. \ref{fig:Price_March}- Fig. \ref{fig:Price_July}.
\begin{figure}[!t]
	\centering
	\includegraphics[width=\linewidth]{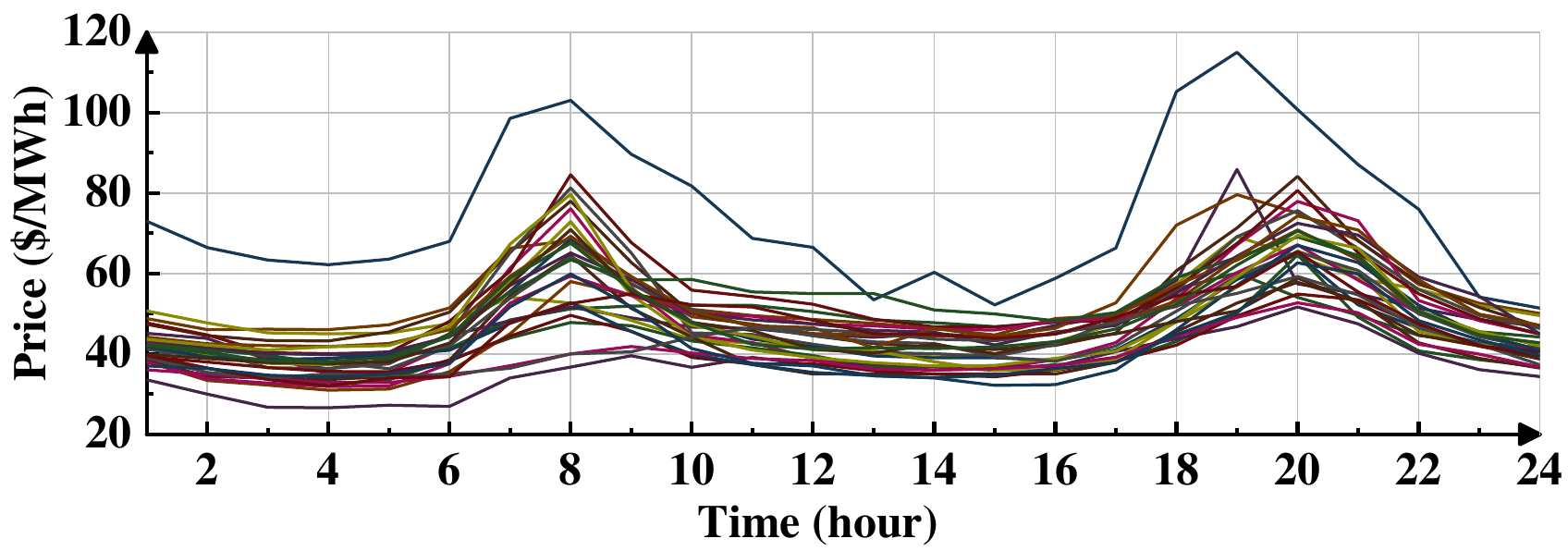}
	\caption{Electricity price trajectories in March 2025.}
	\label{fig:Price_March}
\end{figure}

\begin{figure}[!t]
	\centering
	\includegraphics[width=\linewidth]{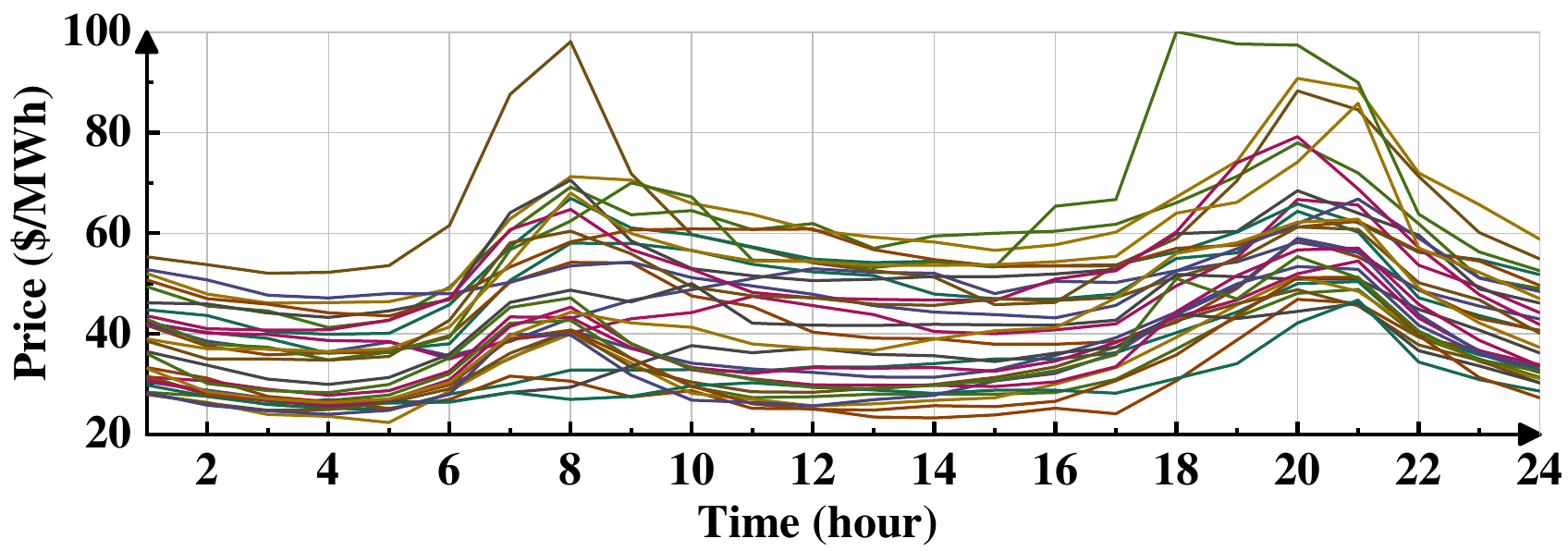}
	\caption{Electricity price trajectories in April 2025.}
	\label{fig:Price_April}
\end{figure}

\begin{figure}[!t]
	\centering
	\includegraphics[width=\linewidth]{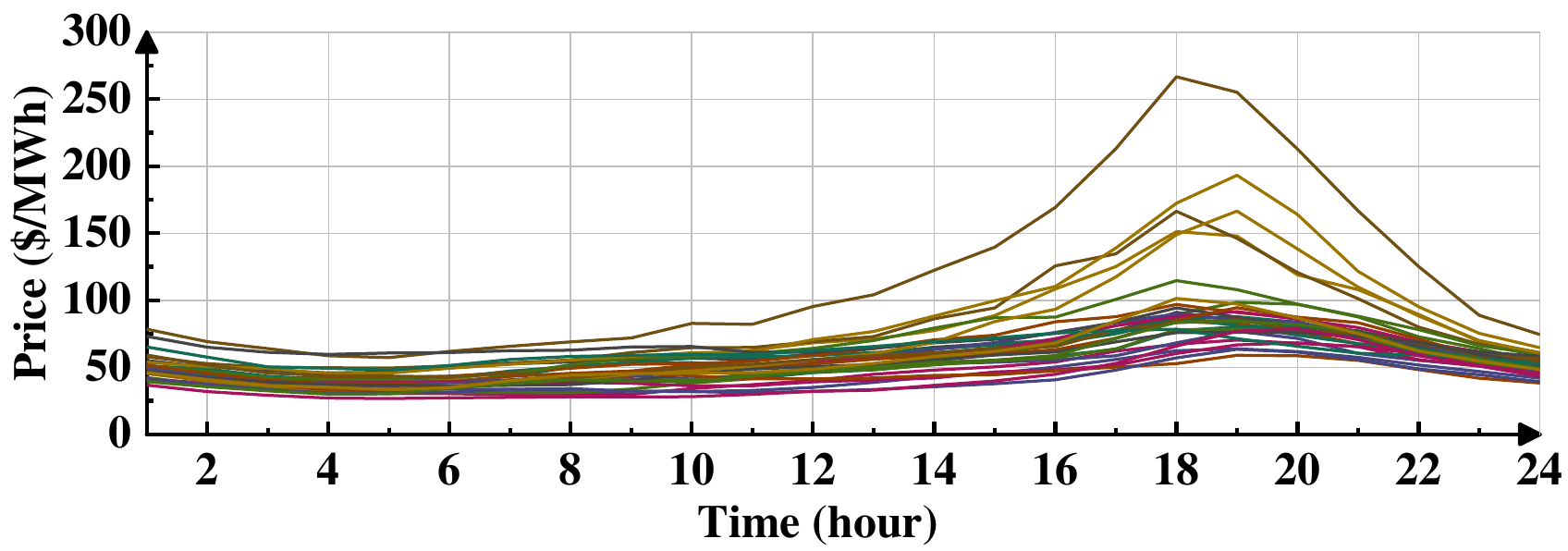}
	\caption{Electricity price trajectories in July 2025.}
	\label{fig:Price_July}
\end{figure}